\newcommand{\columnsversion}[2]{#1}{} 

\columnsversion{
\documentclass[11pt,letterpaper]{article}

\topmargin=-0.35in \topskip=0pt \headsep=15pt \oddsidemargin=0pt
\textheight=9in \textwidth=6.52in \voffset=0in
}{
\documentclass{sig-alternate}
}

\usepackage{amsmath,amsfonts,graphpap,amscd,mathrsfs,graphicx,lscape}
\usepackage{epsfig,amssymb,amstext,xspace}
\usepackage{theorem, bbm}
\usepackage{algorithm,algorithmic}

\usepackage{color}              
\usepackage{epsfig}

\newcommand{\E}{\mathbb{E}}
\newcommand{\abs}[1]{\vert{#1}\vert}

\newtheorem{theorem}{Theorem}[section]
\newtheorem{lemma}[theorem]{Lemma}

\newtheorem{corollary}[theorem]{Corollary}
\newtheorem{defn}[theorem]{Definition}

\newtheorem{fact}[theorem]{Fact}

\newcommand{\algspace}{\text{    } \quad  \text{    }}

{\theorembodyfont{\rmfamily} }

\def\squareforqed{\hbox{\rule{2.5mm}{2.5mm}}}

\def\QED{\ifmmode\squareforqed 
  \else{\nobreak\hfil   
    \penalty50                 
    \hskip1em                  
    \null                      
    \nobreak                   
    \hfil                      
    \squareforqed              
    \parfillskip=0pt           
    \finalhyphendemerits=0     
    \endgraf}                  
  \fi}

\def\blksquare{\rule{2mm}{2mm}}
\def\qedsymbol{\blksquare}
\newcommand{\bg}[1]{\medskip\noindent{\bf #1}}
\newcommand{\ed}{{\hfill\qedsymbol}\medskip}

\columnsversion{
\newenvironment{proof}{\bg{Proof : }}{\ed}
}{}
\newenvironment{proofof}[1]{\bg{Proof of #1 : }}{\ed}

\newcommand{\R}{\ensuremath{\mathbb R}}

\newcommand{\comment}[1]{}
{}  

\newcommand{\junk}[1]{}

\newlength{\tmp} \newlength{\lpsx} \newlength{\lpsy} \newlength{\upsx}
\newlength{\upsy}

\columnsversion{
\newcommand{\email}[1]{\texttt{#1}}
}{}

\newcommand{\one}{\mathbbm{1}}

\begin{document}

\columnsversion{

\setcounter{page}{0}

\title{Polyhedral Clinching Auctions and the Adwords Polytope}

\author{
Gagan Goel\\
       Google Inc., New York\\
       \email{gagangoel@google.com}
\and
Vahab Mirrokni\\
       Google Inc., New York\\
       \email{mirrokni@google.com}
\and
Renato Paes Leme\thanks{This author's contribution to this work was made while
he was an intern at Google Inc. He is currently a graduate student at Cornell
University supported by a Microsoft Fellowship.
}\\
       Cornell University\\
       \email{renatoppl@cs.cornell.edu}
}

}{
\conferenceinfo{STOC'12,} {May 19--22, 2012, New York, New York, USA.} 
\CopyrightYear{2012} 
\crdata{978-1-4503-1245-5/12/05} 
\clubpenalty=10000 
\widowpenalty = 10000

\title{Polyhedral Clinching Auctions and the Adwords Polytope}

\numberofauthors{3}

\author{
\alignauthor
Gagan Goel\\
       \affaddr{Google Research, NYC}\\
       \email{gagangoel@google.com}
\alignauthor
Vahab Mirrokni\\
       \affaddr{Google Research, NYC}\\
       \email{mirrokni@google.com}
\alignauthor
Renato Paes Leme\titlenote{This work was done while this author was an intern at
Google Inc. During the academic year, he is supported by a Microsoft Research
Fellowship.}\\
       \affaddr{Cornell University}\\
       \email{renatoppl@cs.cornell.edu}
}

}

\maketitle

\begin{abstract}

A central issue in applying auction theory in practice is the problem
of dealing with budget-constrained agents. A desirable goal in 
practice is to design
 incentive compatible, individually rational, and Pareto optimal auctions while
 respecting the budget constraints. Achieving this goal is particularly
challenging in the 
 presence of nontrivial combinatorial constraints over the set of feasible
allocations. 

Toward this goal and motivated by AdWords auctions, 
we present an auction for {\em polymatroidal} environments satisfying
the above properties.
Our auction employs a novel clinching technique with
a clean geometric description and only needs an
oracle access to the submodular function defining the polymatroid.
As a result, this auction not only simplifies and generalizes all previous
results, 
it applies to  several new applications including AdWords Auctions, bandwidth
markets, 
and video on demand. In particular, our  characterization of the
AdWords auction as polymatroidal constraints might  be of
independent interest. This allows
us to design the first mechanism for Ad Auctions taking into account
simultaneously budgets, multiple keywords and multiple slots.

We show that it is impossible to extend this result to generic polyhedral
constraints. This also implies an impossibility result for
multi-unit auctions with
decreasing marginal utilities in the presence of budget constraints.

\end{abstract}

\columnsversion{
\renewcommand{\thepage}{}
\clearpage
\pagenumbering{arabic}
}{



}

\section{Introduction}
A large part of auction theory deals with the problem of
designing truthful mechanisms for quasi-linear settings. For these settings,
 the VCG mechanism, or variants of it like affine-maximizers, can be applied to get  optimal or near-optimal
auctions. However, when we deviate from the quasi-linear model, very little is
known. One of the most natural and practically important
feature missing from the quasi-linear model is 
the presence of \emph{budget constraints}. Budgets play a major role
in several real-world auctions where the magnitude of the transactions involved naturally put a financial constraint on the bidders.
Examples of such auctions include those used for
the privatization of public assets in eastern Europe, or those for the distribution of radio spectra
in the US (for a discussion on this, see Benoit and Krishna
\cite{benoit_krishna}). Another important example is that of Ad Auctions
where advertisers explicitly declare budget constraints. In such settings,
 respecting the declared budget constraints is a necessary property
any mechanism must satisfy. There is much
discussion on the source of budget constraints (we refer to
Che and Gale \cite{che_gale} for a detailed discussion on this topic).

Satisfying budget constraints while keeping  incentive compatibility and
efficiency is a challenging problem, and it becomes even harder
in the presence of complex combinatorial constraints over 
the set of feasible allocations. 
In the presence of budgets, 
individual rationality and truthfulness cannot be satisfied at
the same time as  maximizing social welfare~\cite{dobzinski_budgets}, and thus 
the goal of maximizing efficiency can be achieved mainly through 
{\em Pareto-optimal} auctions~\footnote{An auction is
Pareto-optimal if it outputs an allocation and payments such that no alternative set of
allocation and payments improves the utility of at least one agent and keeps
the other agents at least as happy as before. Here agents include bidders and
the auctioneer where the auctioneer's utility is its revenue}. 
Therefore, a desirable goal 
under budget constraints is to design incentive-compatible (IC) and individually-rational (IR)
 auctions while producing Pareto-optimal outcomes.
The first successful example of such mechanisms was developed in the seminal paper of
Dobzinski, Lavi and Nisan~\cite{dobzinski_budgets}, where the authors adapt
the clinching auction framework of Ausubel \cite{Ausubel_multi} to give a
truthful mechanism that achieves Pareto-optimality. Their setting, however,
captures only a simple allocation constraint: 
there is a limited supply of $k$ items and each player has a value of $v_i$ for
each item (and hence value of $v_i \cdot t$ for getting $t$ items) and budget
$B_i$.

As for more general allocation constraints, 
there have been  a couple of subsequent work capturing 
special families of allocation constraints,  
e.g., unit demands~\cite{Aggarwal09}, or  multi-unit demands with matching
constraints \cite{fiat_clinching}.
Although these results are mainly based on Ausubel's clinching auction, each of them need to develop independent techniques to deal with their specific
environment - for example, Fiat, Leonardi, 
Saia, and Sankowski~\cite{fiat_clinching} define a complex
clinching procedure based on trading paths in a bipartite graph and computing
{\em $S$-avoid matchings}.
Our goal in this paper is to extend these results to a much more general class
of polyhedral constraints.
In particular, we would like to understand for what polyhedral environments one can design such auctions, and 
also identify simple environments for which designing such auctions is not
possible.\\

{\bf Our Results and Techniques.}
Firstly, inspired by an application in Sponsored Search Ad Auctions and
several other applications,  we study {\em polymatroid} constraints over feasible allocations
and give an auction that achieves all the desired properties, i.e., it satisfies IC, IR, and produces Pareto-optimal outcomes while satisfying the budget constraints. 
We assume that the budgets are 
public - which was shown in \cite{dobzinski_budgets} to be a necessary
assumption\footnote{Dobzinski et. al. ~\cite{dobzinski_budgets}
showed that with private budgets, truthfulness and Pareto-optimality 
cannot be achieved using deterministic mechanisms - not even for multi-unit
auctions}.
While following Ausubel's framework to design this auction, we need to invent 
the main component of the mechanism, i.e., the {\em clinching step}
that copes with the polyhedral allocation constraints.
Our clinching step uses submodular minimization as a subroutine and only needs a value oracle
access to the submodular function corresponding to the polymatroid. As a result, our mechanism has a
clean geometric description that abstracts away the combinatorial
complications of previous designs. This leaves the auctioneer free to
focus on modeling the environment, and then use our
mechanism as a black-box. 
 This general technique not only generalizes
(and simplifies) the previously known results like  multi-unit auctions with matching
constraints~\cite{dobzinski_budgets,fiat_clinching}, but also extend clinching auctions to
many other applications like the AdWords Auction and settings like spanning
tree auctions and
video on demand~\cite{Bikhchandani11}. 
Our main application is in  sponsored search auctions where we model the AdWords
Auction 
with multiple keywords and multiple position slots per keyword as a polymatroid  called
the \emph{AdWords polytope} (See Section~\ref{subsec:adwords_polytope} for
details). 

In order to extend this result to more general polyhedral constraints,
we turn our attention to $2$-player auctions with budget constraints and
prove several structural properties of Pareto-optimal truthful auctions for polyhedral environments
In particular, we present a characterization of such auctions that results 
in  various impossibility results and one positive result.
On the positive side,  
we present a truthful individually rational Pareto-optimal auction
for any environment if only one player is budget-constrained.
On the other hand,  if more than one player is budget-constrained, 
we illustrate simple polytope constraints for which it is impossible to achieve
a truthful Pareto-optimal auction even for two players. Moreover, as a
byproduct of this characterization, we get an impossibility result for multi-unit auctions with
decreasing marginal utilities. This impossibility result
disproves an implied conjecture by Ausubel
\cite{Ausubel_multi} which has been reinforced by follow-up
papers \cite{dobzinski_budgets,Lavi_May}.
In fact, this conjecture was reinforced by
the fact that getting such an auction is possible whenever the marginals are
flat. In appendix
\ref{appendix:clinching-counter-example} we provide an explicit counter-example
for this case.\\


{\bf Applications to Sponsored Search.}
Online advertisement is a growing business that was worth 25 billion dollars in
2010. It also has become a central piece in
the current internet landscape, since it is the primary way
internet companies monetize
their services. Large
part of this revenue comes from the search advertisement, hence it is not
surprising that is has been extensively studied in the literature. The basic
model of sponsored search ad auctions was proposed simultaneously by Edelman, Ostrovky and
Schwarz \cite{edelman07sellingbillions} and Varian
\cite{Varian06positionauctions}. The authors model the current auction as a
non-truthful mechanism and analyze its equilibrium properties. The social
welfare of such equilibria were studied in \cite{PLT10, Lucier10,
caragiannisetal} and its revenue properties in \cite{LPLT11}.

In search advertising, there are usually multiple keywords and each keyword
has multiple slots associated with it. Most of the previous
work treat auctions for different keywords as being independent, and therefore
focus on a single keyword. A recent paper by Dhangwatnotai \cite{Dhangwatnotai}
approaches the problem of analyzing keyword auctions for multiple keywords, but
it is restricted to a special case of one slot per page. Similarly, the work
of Fiat et al \cite{fiat_clinching} can be seen as an auction for multiple
keywords with only one slot per page.

To the best of our knowledge, our work is the first work to combine multiple slots
per page and multiple keywords. We do so by giving a non-trivial
characterization of the set of all feasible allocations of clicks: we call it the
\emph{AdWords Polytope}. We show it has the structure of a
polymatroid and therefore we can use our result to generate a truthful
Pareto-optimal auction for this setting.\\




{\bf Other Related Work.}
There are two streams of related work. The first, like ours, is on designing
truthful mechanisms when players
have budget constraints with the goal of achieving Pareto-optimal outcomes; for example, 
 \cite{dobzinski_budgets}, and \cite{fiat_clinching}.
Bhattacharya et al \cite{Bhattacharya10} show a budget-monotonicity property
for the clinching auction of \cite{dobzinski_budgets}, therefore arguing that no player can improve his
utility by under-reporting his budget. For the case of unit-demand players, 
Aggrawal et al \cite{Aggarwal09}
design auctions for unit-demand players with budget constraints. 

On the question of maximizing revenue, Borgs et al \cite{Borgs05} gave a
truthful auction whose revenue is asymptotically within a constant factor of the optimal revenue.
These results were improved by Abrams \cite{Abrams06}. Subsequently, Hafalir,
Ravi and Sayedi \cite{Sayedi09} relax the truthfulness requirement, moving to
ex-post Nash equilibrium as a solution concept, and give an auction that, in
equilibrium, has good efficiency and revenue properties. More recently, Pai and
Vohra \cite{pai_vohra} gave a revenue-optimal auction for the Bayesian version
of the problem. We would like to highlight that the above work focused
on the multi-unit setting only.

Our auction also generalizes the ascending auction of Bikhchandani et al
\cite{Bikhchandani11}. The authors consider environments where the set of
allocations is defined by a polymatroid, but don't consider budget constraints.

The second line of related work relates to the practical problem of designing
mechanisms for Ad Auctions. The work of Feldman et al
\cite{feldman08} design an auction for the environment with one keyword and
multiple slots. Their model is, however, different from the standard utilitarian
utility model. Instead of being profit maximizers, the players are clicks
maximizers, i.e., the players want to get as many clicks as possible without
exhausting their budget and without paying more per click than their value,
which is a simpler setting than ours. In order to design their auction, they
describe the structure of the set of possible randomized allocations of players
to slots. We note that the structure they identify is in fact a polymatroid and
use this fact to apply our auction to this setting. We further extend this
characterization to the setting with multiple keywords.

Also for one keyword and multiple slots, Ashlagi et al \cite{Ashlagi} design
an auction for the usual utility model but relax
the truthfulness requirement and get an auction that is Pareto-optimal for all
ex-post Nash equilibria. The main weakness in the setting of \cite{Ashlagi} is 
that the agents are allowed to be allocated only to one slot position for all
the different queries of the given keyword. However, in reality, agents can be
allocated to different slot positions for different queries of a given keyword. 
In the restricted setting of \cite{Ashlagi}, the Pareto-optimality requirement
becomes easier to satisfy.

Independently of our work, Colini-Baldeschi et al \cite{henzinger11} also study
the
problem of designing incentive compatible, individually rational, budget
feasible and Pareto-optimal auctions for sponsored search. The authors present
two auctions satisfying those properties: one for the case with a single keyword
but
multiple slots with different click-through-rates and one for the case of
multiple keywords and multiple slots with homogeneous click-through-rates (i.e.
all slots are identical). 
\\

{\bf  Impossibility results.} 
The impossibility of a Pareto-optimal auction for heterogeneous goods in the
budgeted setting was given in \cite{fiat_clinching}. It remained an open
problem whether an auction was possible if goods where identical, i.e.,
utilities depended only on the number items acquired and not on which items
they were. A very recent result by Lavi and May \cite{Lavi_May} shows an
impossibility result for the case where the valuation can be an arbitrary
function of the number of items - i.e. players are allowed to express
complementarities. Since their setting is more expressive, an impossibility result is
easier. Our impossibility result for multi-unit auctions can be seen as a
stronger version of their result, since we allow players only to express
valuations with diminishing marginals. This came as a surprise to us, since it
was generally believed that such a positive result could be achieved using a
variation of~\cite{dobzinski_budgets}.\\

\comment{
{\bf  Roadmap.} 
In Section \ref{sec:setting_polyhedral_clinching} we present the setting of
polyhedral environments and describe the desideratum in the first part. The rest
of this section described many applications, focusing on the application to
sponsored search. In subsection \ref{subsec:adwords_polytope} we define and
describe the AdWords polytope. Section \ref{sec:clinching_main} we...
}

\columnsversion{\newpage}{}
\section{Auctions for Polyhedral \columnsversion{}{\\}
Environments}\label{sec:setting_polyhedral_clinching}

Consider $n$ players, where player $i$ has a positive value $v_i$ per unit of
some good $g$ 
and a budget of $B_i$. We assume that the valuations are private information of
the players, whereas the budgets are public. We are also given a subset $X
\subseteq \R_+^n$ that defines all the possible ways to allocate the good $g$.
We assume that the subset $X$ is a convex set that is bounded and downward
closed\footnote{This is without loss of genearality if one is allowed to
randomize between outcomes, and can allocate lesser amount of the good $g$ to
any player.}. We will call this set $X$ an {\em environment}. Note that if
player $i$ receives $x_i$ amount of good $g$ and pays $p_i$, her utility $u_i$
is equal to $v_i x_i - p_i$ if $p_i \leq B_i$ and $-\infty$ otherwise. However,
since we will require the mechanism to never charge more than the
budgets, we won't have to deal with the latter case.
Our goal is to design an auction mechanism that elicits valuations $v$ from the
players and outputs 
a feasible allocation
$x(v) \in X$ and a feasible payment vector $p(v) \leq B$
that satisfies the following three properties:  
\begin{itemize}
\item {\em Individual Rationality} (a.k.a. voluntary participation): Each player
has net non-negative utility 
from participating in the auction, i.e., $u_i \geq 0$. 
\item {\em Incentive compatibility} (a.k.a. truthfulness) : It is a
dominant strategy for each player to participate in the auction and report their
true value, i.e., \columnsversion{$v_i x_i(v_i, v_{-i})
- p_i(v_i, v_{-i}) \geq v_i x_i(v'_i, v_{-i}) - p_i(v'_i, v_{-i})$.}{$$v_i
x_i(v_i, v_{-i})
- p_i(v_i, v_{-i}) \geq v_i x_i(v'_i, v_{-i}) - p_i(v'_i, v_{-i}).$$} The
 characterization of single-parameter truthful mechanisms in
\cite{myerson-81, archer01} states that this is equivalent to $x_i$ being a
non-decreasing function of $v_i$ (for a fixed
$v_{-i}$) and payments being calculated by
\columnsversion{
$p_i(v_i, v_{-i}) = v_i x_i(v_i, v_{-i}) - \int_0^{v_i} x_i(u,v_{-i}) du$.}
{$$p_i(v_i, v_{-i}) = v_i x_i(v_i, v_{-i}) - \textstyle\int_0^{v_i}
x_i(u,v_{-i}) du.$$}
\item {\em Pareto-optimality}: An allocation $x(v) \in X$ and payments $
p(v)\leq B$ is Pareto-optimal if and only if there is no alternative allocation
and payments where all players' 
utilities and the revenue of the auctioneer do not decrease,
and at least one of them increases. In
other words, there is no alternative $(x',p')$ such
that $v_i x'_i - p'_i \geq v_i x_i(v) - p_i(v)$, $\sum_i p'_i \geq \sum_i
p_i(v)$ and at least one of those inequalities is strict.
\end{itemize}

Next we prove a useful lemma about the structure of Pareto-optimal outcomes.

\begin{lemma}\label{charac-lemma}
 A feasible outcome $(x,p)$, i.e. $x\in X$ and $p \leq B$, is Pareto-optimal iff
there is no $d \in \R^n$ in a dominated direction at $x$ (i.e. $x+ d \in
X^0 := \{x' \in X; \exists \hat{x} \in X \setminus x', \hat{x} \geq x'\}$)
such that
$d^t v \geq 0$ and $d_i \leq 0$ for all $i$ that have $p_i = B_i$.
\comment{
 A feasible outcome $(x,p)$, i.e. $x\in X$ and $p \leq B$, is Pareto-optimal iff
\begin{enumerate}
\item $x$ is on the Pareto-boundary of $X$.
\item If there exists a feasible direction $d \in \R^n$, i.e. $x+d \in X$, such
that $d.v = 0$ and  $x+d$ is not on the Pareto-boundary of $X$, then $\exists i$
such that $d_i > 0 $ and $p_i = B_i$. 
\end{enumerate}
}
\end{lemma}

\begin{proof}
 In order to show the $\Rightarrow$ direction, assume there is a dominated
direction $d$ such that $d^t v \geq 0$ and $d_i \leq 0$ for all $i$ that have
$p_i = B_i$. Then define $x'_i = x_i+d_i$ and $p'_i = p_i + v_i d_i$
and we obtain same utilities and the total payment didn't decrease, since
$\sum_i p'_i - \sum_i p_i = d^t v \geq 0$. Now, since $x+d$ is not in the
boundary of the polytope, we can give some more of good $g$ to some players
without charging extra payments and increase their utility. Therefore $(x,p)$
is not Pareto-optimal.

For the $\Leftarrow$ direction, suppose $(x',p')$ is a Pareto
improvement. Define $d = x'-x$. First we claim that $d^t v > 0$. By the
definition of Pareto optimality, $v_i x_i - p_i \leq v_i x'_i - p'_i$,
$\sum_i p_i \leq \sum_i p'_i$ and at least one inequality is strict. Summing
them all, we get that $\sum_i v_i x_i < \sum_i v_i x'_i$, which implies that
$d^t v > 0$. Now, consider two cases:

If $x'_i \leq x_i$ for all $i$ with $p_i = B_i$, then $d_i \leq 0$ for all
such $i$. Simply pick some $i$ for which $d_i > 0$ and decrease
$d_i$ slightly. The result will be a dominated direction $d'$ (since
$x+d' \leq x+d$ and $x+d' \neq x+d$) with ${d'}^t v > 0$ and $d'_i \leq 0$ for
all $i$ with $p_i = B_i$.

If $x'_i > x_i$ for some $i$ with $p_i = B_i$. Then define $d'$ such that $d'_i
= d_i$ if $p_i < B_i$ and $d'_i = \min\{0, d_i\}$ if $p_i = B_i$. Now, consider
$x'' = x + d'$ and $p'' = p'$. Clearly $d'$ is a dominated direction (since
$x+d' \leq x+d$ and $x+d' \neq x+d$) and $d'_i \leq 0$ for $p_i = B_i$ by
definition. Now, we will show that ${d'}^t v \geq 0$. Notice that $\sum_i p''_i
= \sum_i p'_i \geq \sum_i p_i$. Also, except for $i$ with $x'_i > x_i$ and $p_i
= B_i$ we have: $v_i x''_i - p''_i = v_i x'_i - p'_i \geq v_i x_i - p_i$. For
theremaining $i$, one has: $v_i x''_i - p''_i = v_i x_i - p'_i \geq v_i x_i -
p_i$. Summing all those inequalities, we get $\sum_i v_i x''_i \geq \sum_i v_i
x_i$, implying that ${d'}^t v \geq 0$.
\end{proof}

\comment{
One implication of the above characterization is that, unlike welfare or revenue efficiency, 
Pareto-optimality is not a convex property, i.e., given two outcomes $(x,p)$ and
$(x',p')$ which are Pareto-optimal, a convex combination $(\alpha x + (1-\alpha)
x', \alpha p + (1-\alpha) p')$ is  not necessarily Pareto-optimal. This is easy
to see by Lemma \ref{charac-lemma} that Pareto-optimal outcomes must lie in the
boundary of $X$, and in general a linear combination of points in the boundary
is not in the boundary. This is not the only problem, though, even in the case
$X = \{x \in \R^2_+; x_1 + x_2 \leq 1\}$, a convex combination might not be
Pareto-optimal. In appendix \ref{appendix:non_convex}, we further develop
this point by showing that the family of mechanisms in Dobzinski et al
\cite{dobzinski_budgets} do not form a convex set. This highlights the
difficulty of designing mechanisms with budgets: unlike the quasi-linear
settings, searching for a suitable mechanism is a search in a non-convex space.
}

Another simple observation is that if $(x,p)$ is a Pareto-optimal outcome in which no
budget is fully exhausted, then $x  = \text{argmax}_{x \in X} v^t x$. For
small valuations, any Pareto-optimal mechanism that satisfies individually
rationality cannot exhaust budgets, 
so it must behave like VCG.

\subsection{Polymatroidal environments}\label{subsec:adwords_polymatroid}

Our most interesting applications correspond to settings where the environment is a
packing polytope $P = \{x \in \R^n_+; Ax \leq b\}$ for some $m \times n$ matrix
with $A_{ij} \geq 0$ and $b \in \R^m_+$. We call such environments a polyhedral
environment. Examples of polyhedral environments are ubiquitous in game theory
(see \cite{nguyen07, NV11} for many examples).

A rich subclass of packing polytopes is the class of polymatroids, which are
polytopes that can be written as $P = \{x \in \R^n_+; \sum_{i \in S} x_i \leq
f(S)\}$ where $f:2^{[n]} \rightarrow \R_+$ is a monotone submodular function,
i.e., a function satisfying:
$$f(S \cup T) + f(S \cap T) \leq f(S) + f(T), \forall S,T \subseteq [n]$$
$$f(S) \leq f(T), \forall S \subseteq T \subseteq [n]$$
Such polymatroidal environments generalize  matroid environments.
It is easy to see that all previously studied settings are instances of a
these environments: Dobzinski et al's result \cite{dobzinski_budgets}
corresponds
to the \emph{uniform matroid} and Fiat et al's result \cite{fiat_clinching}
corresponds
to the \emph{transversal matroid}.  Bikhchandani et al \cite{Bikhchandani11}
give many examples of polymatroid environments  including scheduling with due
dates, network 
planning, pairwise kidney exchange, spatial markets, bandwidth markets and 
multi-class queueing systems \cite{Bikhchandani11}. In section
\ref{sec:applications} we discuss some of those applications in more depth and
present a novel application of polymatroids to sponsored search auctions.

\section{Clinching Auction for \columnsversion{}{\\}
Polymatroids}\label{sec:clinching_main}

In this section, we describe our main positive result, i.e.,  an auction with 
all the desirable properties for polymatroidal environments. This auction 
 is based on the clinching auctions framework of Ausubel
\cite{Ausubel_multi}. \comment{
Although following Ausubel's framework, we need to invent 
the main component of the mechanism, i.e., the {\em clinching step}
in a way that copes with the allocation constraints.
As we will see, this auction is well-defined for all polyhedral
constraints $P$, but only for polymatroids it is guaranteed to always produce
Pareto-optimal outcomes. Moreover, we need the polymatroidal structure to make
the most critical step (clinching) polynomial-time computable.
}
Before we study more complicated constraints, let's
recall the clinching auction \cite{Ausubel_multi, dobzinski_budgets} for the
multi-unit setting, i.e., $P = \{x; \sum_i x_i \leq s_0\}$. We begin by
setting the supply $s = s_0$ and $B_i$ the budget available to each agent. We
maintain a price clock $p$ that begins at zero and gradually ascends. For each
price $p$, the agents are asked how much of the good they demand at the current
price. Their demand will be $d_i = \frac{B_i}{p}$ (how much they can afford
with their remaining budget) if $p \leq v_i$, and zero otherwise (the case
where the price exceeds their marginal value). Then agent $i$ is able to
\emph{clinch} an amount $\delta_i = [s - \sum_{j\neq i} d_j]^+$, which is the
minimum amount we can give to player $i$ while we are still able to meet the
aggregate demands of the other players. \emph{Clinching} means that player $i$
gets $\delta_i$ amounts of the good, and $\delta_i p$ is subtracted from his
budget. The price increases and we repeat the process until the supply is
completely sold.

The heart of the mechanism is the clinching step and generalizing it for  more
complicated environments involves various challenges: how does one
define the notion of supply and aggregate demand \comment{from a set of agents}
(it is not a
single number anymore, since there are constraints restraining the
possible allocation)? Finally, we need to make sure the clinching step doesn't
violate feasibility. \\

{\bf Clinching Framework.}
First, in Algorithm \ref{polyhedral-clinching-auction}, we 
consider a slightly modified
version of the clinching framework: we maintain a price vector $p \in \R^n_+$
and increase the prices one player at a time. The vector $\rho \in \R^n_+$
contains the promised allocations in each step and its final value is the
final allocation of the mechanism. The payment of each agent is the total
amount that was deducted from their budget during the execution~\footnote{Note
that the details of the main procedure {\bf clinch} is not described in this algorithm.}.

\begin{algorithm}[h]
 \caption{Polyhedral Clinching Auction}
  \textbf{Input:} $P, v_i, B_i$\\
  $p_i = 0$, $\rho_i = 0$, $\hat{i}=1$ \\
  \textbf{do} \\
  \algspace $d_i = B_i / p_i$ if $p_i < v_i$ and $d_i = 0$ otherwise, \\
  \algspace $\delta = \text{\textbf{clinch}}(P,\rho,d)$, \\
   \algspace $\rho_i = \rho_i + \delta_i$, \quad $B_i = B_i - p_i \delta_i$, \\
   \algspace  $d_i = B_i / p_i$ if $p_i < v_i$ and $d_i = 0$ otherwise, \\
  \algspace $p_{\hat{i}} = p_{\hat{i}} + \epsilon$ ,\quad $\hat{i} = \hat{i}
+ 1 \mod n$ \\
  \textbf{while} $d \neq 0$ 
\label{polyhedral-clinching-auction}
\end{algorithm} 

For each price, we
calculate the demand $d_i$ of each player, which is the amount of the good they
would like to get for price $p$. Then we invoke a procedure called
\textbf{clinch} which decides the amount to grant to each player at that price.
We update the promises, remaining budget and adjust demands\footnote{Clearly
updating demands is not necessary at this point, but we do in order to make
the analysis cleaner.}. Then we increase the price.

In order to define clinching, we need to define analogues of the remnant supply
and to demands for the case where the the environment is a generic polytope.
Instead of being a single number as in the multi-unit auctions case, the
remnant supply and aggregate demands will be polytopes:

\begin{defn}[aggregate demands]\label{remnant_supply_defn}
Given $P$, a vector of promised allocation $\rho \in P$, the remnant supply is
described by the polytope $P_\rho = \{x \geq 0; \rho + x \in P\}$. If $d \in
\R^n_+$ is the demand vector, the aggregate demand is defined by $P_{\rho,d} =
\{x \geq 0; \rho+x \in P, x\leq d\}$.
\end{defn}

In the multi-unit auctions case, the amount player $i$ clinched was the maximum
amount we could give him while still being able to meet the demands of the
other players. We generalize this notion to polyhedral environments (the
concepts are depicted in Figure \ref{fig:clinching}):

\begin{defn}[polyhedral clinching]\label{polyhedral_clinching_dfn}
The {\em demand set} of
players $[n]\setminus i$ if one allocates $x_i$ to player $i$ is represented by
the polytope $P_{\rho,d}^i(x_i) = \{x_{-i} \in \R_+^{[n]\setminus i}; (x_i,
x_{-i}) \in P_{\rho,d}\}$. Since $P$ is a
packing polytope, clearly $P_{\rho,d}^i(x_i) \supseteq P_{\rho,d}^i(x'_i)$ if
$x_i \leq x'_i$. The amount player $i$ is able to \textbf{clinch} is the
maximum amount we can give him without making any allocation for the other
players infeasible. More formally,  $\delta_i = \sup\{x_i \geq 0;
P_{\rho,d}^i(x_i) = P_{\rho,d}^i(0)\}$.
\end{defn}

\columnsversion{
\begin{figure}
\centering
\includegraphics{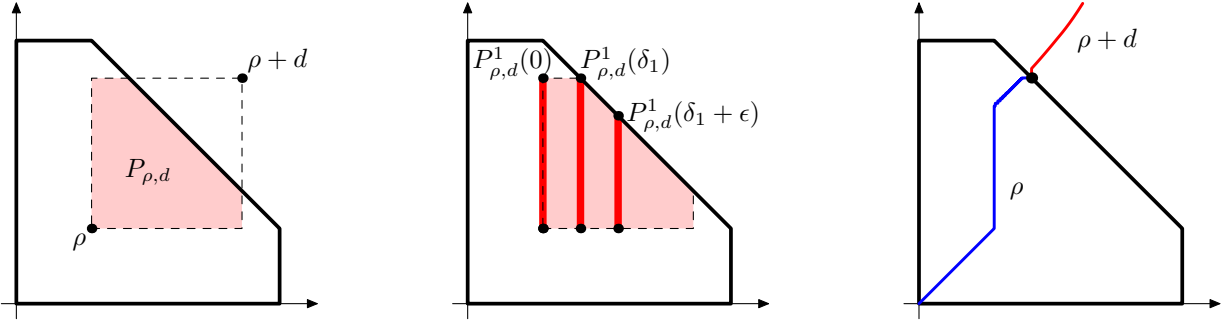}
\caption{Illustration of polyhedral clinching: the first two figures depict
the polytopes defined in Definitions \ref{remnant_supply_defn} and
\ref{polyhedral_clinching_dfn}. The third depicts the mechanism running on
polytope $P$: during the execution, the vector $\rho$ walks inside the polytope
(blue line) and the vector $\rho+d$ walks outside it (red line), until they meet
at the boundary. The point they meet corresponds to the final allocation.}
\label{fig:clinching}
\end{figure}
}
{
\begin{figure*}
\centering
\includegraphics{pictures/polymatroid_clinching1.mps}
\caption{Illustration of polyhedral clinching: the first two figures depict
the polytopes defined in Definitions \ref{remnant_supply_defn} and
\ref{polyhedral_clinching_dfn}. The third depicts the mechanism running on
polytope $P$: during the execution, the vector $\rho$ walks inside the polytope
(blue line) and the vector $\rho+d$ walks outside it (red line), until they meet
at the boundary. The point they meet corresponds to the final allocation.}
\label{fig:clinching}
\end{figure*}
}

We need to ensure that the clinching step is well-defined, i.e., that after
clinching is performed, the vector of promised allocations is still feasible.
This is done by the following lemma:

\begin{lemma}\label{lemma_feasibility}
 For each step of the auction above, if $\rho \in P$, then $\rho+\delta \in P$.
\end{lemma}

\begin{proof}
 Let $\chi^i$ be the $i$-th coordinate vector. Note that $\delta_1 \chi^1 \in
P_{\rho,d}$ by definition of $\delta$. Now, notice that  $\delta_1 \chi^1 \in
P^2_{\rho,d}(0) = P^2_{\rho,d}(\delta_2)$, so: $\delta^1 \chi_1 + \delta_2
\chi^2 \in P_{\rho,d}$. By induction, we can show that $\sum_{i=1}^j \delta_i
\chi^i \in P_{\rho,d}$. The induction is easy: $\sum_{i=1}^j \delta_i
\chi^i \in P_{\rho,d}^{j+1}(0) = P_{\rho,d}^{j+1}(\delta_{j+1})$, so
$\sum_{i=1}^{j+1} \delta_i \chi^i \in P_{\rho,d}^{j+1}$. 
\end{proof}

This auction is clearly truthful, since each player $i$ reports only $v_i$, and she
can stop her participation earlier (which she doesn't want, since she will
potentially miss items she are interested in) or later (which will potentially give 
her items
for a price higher than her valuation). It is also individually rational,
since players only get items for prices below their valuation and respect
budgets by the definition. Notice that those facts are true regardless of the
trajectory of the price vector: any process that increases prices (in a
potentially non-uniform way) has this property.

\begin{lemma}\label{almost_all_properties_lemma}
 The auction in Algorithm \ref{polyhedral-clinching-auction} 
 along with the clinching step described in
 Definition~\ref{polyhedral_clinching_dfn} 
 is truthful, individually-rational and budget-feasible.\\
\end{lemma}

{\bf Clinching for polymatroids.}
Notice that we haven't used anything from polymatroids yet, so Lemma
\ref{almost_all_properties_lemma} holds for any polytope $P$. However, two
things are left to be shown: (i)
that amount clinched can be computed efficiently and (ii) that the outcome is
Pareto optimal. To show both
of  these properties, we use the fact that $P$ is a polymatroid.

\begin{lemma}\label{eff-clinching}
 If the environment is a polymatroid $P$ defined by a submodular
function $f$,
then the amount player $i$ clinches in Algorithm
\ref{polyhedral-clinching-auction} is given by: $$\delta_i = ( \textstyle\max_{x
\in
P_{\rho,d}} \one^t x ) - ( \textstyle\max_{x \in P_{\rho,d}} \one_{-i}^t x_{-i}
). $$
Moreover, this can be calculated efficiently using submodular minimization.
\end{lemma}

The main ingredients of the proof are the following two facts about
polymatroids:

\begin{fact}[Schrijver \cite{schrijver-book}, sections
44.1 and 44.4]\label{schrijver-fact} If $P$ is a polymatroid defined by the
submodular
function $f$, then $P_{\rho,d}$ is also
a polymatroid defined by the following submodular function:
$$\hat{f}(S) = \min_{T \subseteq S} \{f(T) - \rho(T) + d(S \setminus T )\}.$$
Notice that $\hat{f}(\cdot)$ might not be monotone. However, $$\bar{f}(S) =
\min_{S' \supseteq S} \hat{f}(S')$$ is a monotone submodular function that
defines the same polymatroid.
\end{fact}

\begin{fact}\label{equal-polymatroids}
Given two monotone submodular functions $f, \tilde{f}$, then the
polymatroids $P, \tilde{P}$ defined by them are equal iff the functions are
equal. The $(\Leftarrow)$ direction is trivial. For the other direction, notice
that
if $f(S) < \tilde{f}(S)$ say for $S = \{1,\hdots, i\}$, notice that the point
$x$ such that $x_j = \tilde{f}(\{1..j\})-\tilde{f}(\{1..j-1\})$ for $j \leq i$
and zero otherwise is such that $x \in \tilde{P} \setminus P$.
\end{fact}

\begin{proofof}{Lemma \ref{eff-clinching}}
Using the fact \ref{schrijver-fact}, we know that $P^i_{\rho,d}(x_i)$ is also a
polymatroid defined over $[n]\setminus i$ by the function $\tilde{f}(S) = \min\{
\bar{f}(S), \bar{f}(S\cup i) - x_i\}$. Now, we use Fact
\ref{equal-polymatroids} to see that $P^i_{\rho,d}(x_i) = P^i_{\rho,d}(0)$ iff
$\bar{f}(S) \leq \bar{f}(S\cup i) - x_i, \forall S \subseteq [n] \setminus i$.
So, $\delta_i = \min_{S \subseteq [n] \setminus i} \bar{f}(S\cup i) -
\bar{f}(S)$. Since $\bar{f}$ is submodular, the smallest marginal can only be
$$\bar{f}([n]) - \bar{f}([n]\setminus i) = \max\{0, \hat{f}([n]) -
\hat{f}([n]\setminus i) \}$$ which is exactly the expression in the
statement of the lemma. Now, one can easily see that 
evaluating $\hat{f}$ is a submodular minimization
problem.
\end{proofof}

Now we  prove that the outcomes are Pareto-optimal in two steps. The first step
is to characterize Pareto-optimal allocations for polymatroids. This
characterization is stronger than that of Lemma \ref{charac-lemma}, since it
explores the structure of polymatroids. Afterwards, we show that the outcomes of
the chinching auction defined in Algorithm
\ref{polyhedral-clinching-auction} satisfy the two conditions in the
characterization lemma.

In the following, for a vector $x \in \R^n$ and $S \subseteq [n]$ we denote
$x(S) = \sum_{i \in S} x_i$.

\begin{lemma}\label{charac-polymatroid}
For a polymatroidal environment  $P$ defined by a submodular function $f$,
 an allocation $(x,p)$ is Pareto optimal iff:
\begin{enumerate}
 \item All items are sold, i.e., $x([n]) = f([n])$, and
 \item Given a player $i$ with $p_i < B_i$ and player $j$ with $v_j < v_i$, then
there exists a set $S$ such that $x(S) = f(S)$, $i \in S$ and $j \notin S$
\end{enumerate}
\end{lemma}

The following elementary facts about submodular functions will be useful in the
proof of the Lemma~\ref{charac-polymatroid}:

\begin{fact}\label{union-intersection-fact}
Given a vector $x \in P$,  if two sets $S,T$ are tight (i.e.
$x(S) = f(S)$ and $x(T) = f(T)$), then $S \cap T$ and $S\cup T$ are also tight.
The proof is quite elementary: $x(S \cup T) = x(S) + x(T) - x(S \cap T) \geq
f(S) + f(T) - f(S\cap T) \geq f(S \cup T)$. So, all the inequalities must be
tight and therefore $x(S \cup T) = f(S \cup T)$ and $x(S \cap T) = f(S \cap T)$.
\end{fact}

\begin{fact}\label{slack-element-fact}
If $x([n]) < f([n])$ then there is one component $x_i$ that we can increase by
$\delta > 0$ such that $x$ is still in $P$. It follows from the previous fact:
if all players $i$ were contained in a tight set, one could take the union of
those and $[n]$ would be tight. Then there is some element $i$ which is in no
tight set.
\end{fact}

\begin{proofof}{Lemma \ref{charac-polymatroid}} The $(\Rightarrow)$ direction
is easy. If $x([n]) < f([n])$ then  we can increase some $x_i$ (Fact
\ref{slack-element-fact}) and still
get point $P$ generating a Pareto improvement. Also, if there is $p_i < B_i$
and $v_j < v_i$ and no tight set separating them, then we can consider another
outcome where we increase $x_i$ by some $\delta > 0$, decrease $x_j$ by some
$\delta < 0$ and still get a feasible point improving $x^t v$. Now, this
would not be Pareto optimal by Lemma \ref{charac-lemma}.

For the $(\Leftarrow)$ direction, let $(x,p)$ be an outcome satisfying
properties $1$ and $2$ and suppose $(x',p')$ is a Pareto-improvement. This
means that $v_i x'_i - p'_i \geq v_i x_i - p_i$ and $\sum_i p'_i \geq \sum_i
p_i$.

Let $\{i_1, \hdots, i_k\} = \{n\} \cup \{i; p_i < B_i\}$, sorted in
non-increasing order of $v_i$. Using property $2$ (notice it holds for player
$n$ trivially) together with fact \ref{union-intersection-fact},
we define the following family of tight sets $S_1 \subseteq S_2 \subseteq \hdots
\subseteq S_k = [n]$, tight in the sense that $x(S_i) = f(S_i)$. For all $v_t <
v_{i_1}$ there is a tight set
$S_{1t}$ that has $i_1$ but not $t$. Let $S_1$ be the intersection of such sets.
Now, given $S_1 \subseteq \hdots \subseteq S_{j-1}$, we define $S_j$ in the
following way: If $i_j \in S_{j-1}$,  take $S_j = S_{j-1}$ (notice can only
happen if $v_{i_j} = v_{i_{j-1}}$). If not, for each $v_t <
v_{i_j}$ there is a tight set $S_{jt}$ that has $i_j$
but not $t$. Now, define $S_j$ as the union of $S_{j-1}$ and the intersection of
the $S_{jt}$ sets. 

By eliminating duplicates and its corresponding elements from $\{i_1
\hdots i_k\}$, we get a family $S_1 \subset \hdots \subset S_k$. Define $T_j =
S_j \setminus S_{j-1}$ and it is clear the family obtained has the following
properties:
\begin{itemize}
 \item all $t \in S_j$ have $v_t \geq v_{i_j}$
 \item $i_j \in T_j$
 \item for all $i \in T_j$ either $v_i = v_{i_j}$ or $p_i = B_i$.
\end{itemize}
Let $T'_j = \{i \in T_j ; v_i \neq v_{i_j}\}$ and $T''_j = \{i \in T_j; v_i =
v_{i_j}\}$. Since the players in $T'_j$ have exhausted their budget, $p_i \geq
p'_i$. Using that and Pareto-optimality, we get:

\begin{equation}\label{chain-ineq-tj}
\begin{aligned}
& \sum_{i \in T'_j} p_i - p'_i \geq \\& \quad \geq \sum_{i \in
T'_j , x_i \geq x'_i} p_i - p'_i \geq \sum_{i \in T'_j , x_i \geq x'_i} v_i(x_i
- x'_i) \stackrel{*}{\geq} \\
& \quad  \geq \sum_{i \in T'_j , x_i \geq x'_i}
v_{i_j}(x_i - x'_i) \stackrel{**}{\geq} 
\sum_{i \in T'_j } v_{i_j}(x_i - x'_i) 
  \end{aligned}
\end{equation}

Now, we can add the inequality $p_{i} - p'_{i} \geq v_{i_j}(x_{i} -
x'_{i})$ for $i \in T''_j$ and obtain:
\begin{equation}\label{ineq-tj}
 \sum_{i \in T_j } p_i - p'_i \geq \sum_{i \in T_j}
v_{i_j}(x_i - x'_i)
\end{equation}
Summing those for all $j$ and get:
\columnsversion{
$$ \begin{aligned}\sum_i p_i - p'_i & \geq \sum_j \sum_{i \in T_j}v_{i_j}(x_i -
x'_i) = \sum_j (v_{i_j}
- v_{i_{j+1}}) \sum_{i \in S_j} (x_i - x'_i) \geq 0 \end{aligned}$$
}{
$$ \begin{aligned}\sum_i p_i - p'_i & \geq \sum_j \sum_{i \in T_j}v_{i_j}(x_i -
x'_i) = \\ & = \sum_j (v_{i_j}
- v_{i_{j+1}}) \sum_{i \in S_j} (x_i - x'_i) \geq 0 \end{aligned}$$
}
since $x(S_j) = f(S_j) \geq x'(S_j)$. Therefore $\sum_i p_i \geq \sum_i p'_i$
and therefore equal. This means in particular all of the inequalities in
(\ref{chain-ineq-tj}) and (\ref{ineq-tj}) must be tight. Therefore for all $i
\in T'_j$ we need to have $x_i = x'_i$, since if $x_i > x'_i$
then inequality $*$ in (\ref{chain-ineq-tj}) would be strict. If $x_i < x'_i$,
then inequality $**$ would be strict.  We use this fact to show that $\sum_{i
\in
S_j} v_i(x_i - x'_i) \geq 0$ by induction on $j$. If we show that, we can take
$j = k$ and then we are done, since this will imply that $\sum_i v_i x_i \geq
\sum_i
v_i x'_i$ and therefore $(x',p')$ cannot be a Pareto-improvement.

For $j=1$, this is trivial, since we can write:
$$ \sum_{i\in S_1, v_i \neq v_{i_1}} v_i(x_i -x'_i) \geq \sum_{i\in S_1, v_i
\neq
v_{i_1}} v_{i_1}(x_i -x'_i) $$ since both terms are zero, and then sum $v_i
(x_i-x'_i)$ for the rest of the
elements in $S_1$ and use the fact that $S_1$ is tight. For other $j$, we use
that:
\columnsversion{
$$
\begin{aligned}
& \sum_{i \in S_j} v_i(x_i - x'_i) \geq \\ & \quad \geq v_{i_{j-1}} \sum_{i \in
S_{j-1}}
(x_i - x'_i) + \sum_{i \in T'_j} v_i(x_i - x'_i) + \sum_{i \in
T''_j} v_{i_j}(x_i - x'_i) \geq \\
& \quad \geq  v_{i_j} \sum_{i \in S_j} (x_i - x'_i) = v_{i_j} (x(S_j) -
x'(S_j)) \geq 0,
\end{aligned}
$$
}{
$$
\begin{aligned}
& \sum_{i \in S_j} v_i(x_i - x'_i) \geq \\ & \quad \geq v_{i_{j-1}} \sum_{i \in
S_{j-1}}
(x_i - x'_i) + \sum_{i \in T'_j} v_i(x_i - x'_i) \\ & 
\quad \quad  + \sum_{i \in
T''_j} v_{i_j}(x_i - x'_i) \geq \\
& \quad \geq  v_{i_j} \sum_{i \in S_j} (x_i - x'_i) = v_{i_j} (x(S_j) -
x'(S_j)) \geq 0,
\end{aligned}
$$
}
by the fact that $S_j$ is tight.
\end{proofof}

Now, we argue that, for sufficiently small $\epsilon$, the outcome satisfied the
two properties in Lemma \ref{charac-polymatroid} and hence is Pareto-optimal. 
We prove this fact using the following
sequence of lemmas:

\begin{lemma}\label{after-clinching}
After the clinching step is executed, and before updating prices, 
$\hat{f}([n]) \leq
\hat{f}([n] \setminus j), \forall j \in [n]$.
\end{lemma}

\begin{proof}
 In the clinching step, given an initial $\hat{f}_0$, we define $\delta_i =
\max \{0, \hat{f}_0([n]) - \hat{f}_0([n]\setminus i) \}$. After we update
$\rho,B,d$,
$\hat{f}$ is updated to $\hat{f}_1(S) = \hat{f}_0(S) - \delta(S)$. Now, it is
easy to check that:
\columnsversion{
$$\begin{aligned} \hat{f}_1([n]) & = \hat{f}_0([n]) - \sum_i \delta_i  =
\hat{f}_0([n]) - \delta_j - \sum_{i \neq j} \delta_i  \leq
\hat{f}_0([n] \setminus j) - \sum_{i \neq j}\delta_i = \hat{f}_1([n] \setminus
j) \end{aligned}$$
}{
$$\begin{aligned} \hat{f}_1([n]) & = \hat{f}_0([n]) - \sum_i \delta_i  =
\hat{f}_0([n]) - \delta_j - \sum_{i \neq j} \delta_i  \leq \\ & \leq
\hat{f}_0([n] \setminus j) - \sum_{i \neq j}\delta_i = \hat{f}_1([n] \setminus
j) \end{aligned}$$
}
\end{proof}

\begin{lemma}\label{sell-all-the-goods}
 The outcome $(x,p)$ of the clinching auction is such that $x([n]) = f([n])$.
\end{lemma}

\begin{proof}
We show the following invariant: if we define $\hat{f}$ as in Fact
\ref{schrijver-fact}, updating it each round as $\rho, d$ changes, we claim that
the value of $\one^t \rho + \hat{f}([n])$ remains constant.

To do so, we consider the events that can cause it to drop:
\begin{enumerate}
\item clinching: just after clinching occurs (i.e. $\rho_i$  increases by
$\delta_i$, budgets decrease by $p \delta_i$, demands
are adjusted, but before the price 
increases), the amount $\one^t \rho + \hat{f}([n])$ remains the same since
$\rho^t \one$ increases by $\delta^t \one$ and for all $S$, $\hat{f}(S)$
decreases by $\delta(S)$, because to each $i$, $\rho$ increases by $\delta_i$
and $d_i$ decreases by $\delta_i$.
\item price $p_i$ increases and $d_i$ decreases by $\theta, 0 \leq \theta \leq
d_i$. If $\one^t \rho + \hat{f}([n])$ decreased then there was some $T, i
\notin
T$ such that:
$$\hat{f}([n]\setminus i) + d_i - \theta \leq f(T) - \rho(T) + d([n]\setminus T)
- \theta < \hat{f}([n])$$
Using Lemma \ref{after-clinching}, we know that $
\hat{f}([n]) \leq \hat{f}([n]\setminus i) $, so $d_i < \theta$ which is not
true.
\end{enumerate}
\end{proof}

The proofs of the previous two lemmas intuitively establishes the maximality of
the clinching procedure. Lemma \ref{after-clinching} can be interpreted as
saying that if we apply the clinching procedure twice, without updating prices,
then the second time will have no effect. The proof of Lemma
\ref{sell-all-the-goods} identifies an invariant that is maintained during the
execution of the mechanism.

\begin{lemma}\label{second-property}
 If $\epsilon < \min_{v_i \neq v_j} \abs{v_i - v_j}$, then property $2$ of
Lemma \ref{charac-polymatroid} is satisfied.
\end{lemma}

\begin{proof}
 Suppose not and for the final outcome there are $v_j < v_i$, $p_i < B_i$ and
all
sets $S$ such that $i \in S$, $j \notin S$ are not tight. First, clearly $x_j
\neq 0$, otherwise $[n] \setminus j$ would be tight by Lemma
\ref{sell-all-the-goods}. Then consider $\tilde{x}$
where $\tilde{x}_i = x_i + \theta$, $\tilde{x}_j = x_j - \theta$ and
$\tilde{x}_k = x_k$ for all $k \neq i,j$. It is feasible for some small
$\theta$.

Now, consider the promised allocation $\rho$ and demands $d$ just before the
last time
player $j$ clinched an amount $\delta_j > 0$. If necessary decrease
$\theta$ so that it
becomes smaller than this last amount clinched, i.e., $\theta < \delta_j$.
 At this point $\rho \leq x \leq \rho + d$. By the definition
of clinching: $P_{\rho,d}^j(\delta_j) = P_{\rho,d}^j(\theta) = P_{\rho,d}^j(0)$.
\comment{ but this is not true,
since $(\tilde{x}-\rho)_{-j} \in P^j_{\rho,d}(0) \setminus
P^j_{\rho,d}(\theta)$.
}

At this point, $\rho \leq x$ and $\rho_j + \theta < \rho_j + \delta_j = x_j$.
Therefore $\tilde{x} \geq \rho$. Also, we have that $x - \rho \leq d$ and $x_i -
\rho_i < d_i$, since agent $i$ hasn't dropped his demand
to zero yet and his demand never increases and won't be met while $v_i < p_i$.
Here we are strongly using that $\epsilon < \min_{v_i \neq v_j} \abs{v_i - v_j}$
to ensure that for the last time player $j$ clinches, player $i$ demand is not
zero yet. This implies that $\tilde{x} -\rho \in P_{\rho,d}$ so $(\tilde{x}
-\rho)_{-j} \in P_{\rho,d}^j(0)$. Now, the fact that $P_{\rho,d}^j(0) =
P_{\rho,d}^j(\delta_i)$ implies that $\hat{x} = (x_j, \tilde{x}_{-j}) \in P$.
But  $\hat{x}([n]) = x([n]) + \theta = f([n]) + \theta > f([n])$, which is an
absurd.
\end{proof}

We can summarize the results as:

\begin{theorem}
 For a polymatroidal environment, the auction in Algorithm
\ref{polyhedral-clinching-auction} along with
 the clinching step described in
 Definition~\ref{polyhedral_clinching_dfn}
has all the desirable properties.\\
\end{theorem}

 \textbf{Extensions and limitations of the clinching framework:} The
clinching framework described in Algorithm \ref{polyhedral-clinching-auction}
and Definition~\ref{polyhedral_clinching_dfn}
is quite flexible: one can change the way clinching is done or the way prices
ascend and obtain an auction that is still truthful, individually rational, and
respects budgets. Pareto-optimality, however, is a delicate property to
achieve. A natural question is for which environments Pareto-optimality is
still achievable? In appendix \ref{appendix:framework_limits}, we show we can
extend this framework a little further (to scaled polymatroids), but not
further than that. \\

\textbf{Faster clinching:} We showed in this section that for a
generic polymatroid, we can calculate the clinched amount using submodular
minimization as a sub-routine. For each individual environment, however, one
can usually find much faster clinching subroutines. We illustrate this in
appendix \ref{appendix:eff_clinching_adwords} for the single-keyword AdWords
polytope.
\section{The AdWords Polytope and \columnsversion{}{\\} other
applications}\label{sec:applications}

We begin by discussing some interesting applications of the auction presented in
section \ref{sec:clinching_main}. Then we introduce a novel application of
polymatroidal constraints to sponsored search auctions that generalizes the
classical models of Edelman et al \cite{edelman07sellingbillions} and Varian
\cite{Varian06positionauctions}.

\begin{itemize}
 \item {\em Multi-unit auctions \cite{dobzinski_budgets}:} corresponds to the
polymatroids associated with constant submodular functions, i.e., $f(S) = Q,
\forall S$.

 \item {\em Combinatorial auctions with matching constraints
\cite{fiat_clinching}:} there is a bipartite graph $([m],[n],E)$ between
items $[m]$ and bidders $[n]$ and each buyer $i$ has additive value $1$ for each
item $j$ such that $(i,j) \in E$ and value $0$ for each item not connected to
him. We can represent this setting by a polymatroid where $f(S)$ is the number
of items connected to some player in $S$. This is called the \emph{transversal
matroid}.

\item {\em Video on demand \cite{Bikhchandani11}:} Consider company that
provides video on demand that is located on a node $s$ of a direct network with
capacities on the edges $G=(V,E,c)$. Each buyer corresponds to a node in the
network. An allocation $x$ is feasible if it is possible to transmit at rate
$x_i$ for each player $i$ simultaneously. This is possible if for each subset $S
\subseteq [n]$ of players, $\sum_{i \in S} x_i$ is smaller then the min-cut from
$s$ to $S$. Using the submodularity of the cut-function, it is easy to see that
the environment is a polymatroid.

\item {\em Spanning tree auctions:} Consider the abstract setting where the
agents are edges of a graph $G$ and the auctioneer is allowed to allocate goods
to a set only if it has no cycles. This corresponds to the graphical matroid of
graph $G$. A more practical setting is when a telecommunication company owns a
network that contains cycles and decides to auction their redundant edges. This
setting corresponds to the dual-graphical matroid of $G$.
\end{itemize}

\subsection{AdWords Polytope}\label{subsec:adwords_polytope}

Consider $n$ advertisers and $m$ keywords. Each
advertiser $i$ is interested in a subset of the keywords $\Gamma(i) \subseteq
[m]$. For a keyword $k$, we denote by $\Gamma(k)$, the set of advertisers
interested in this keyword. With each keyword $k$, we associate
$\abs{\Gamma(k)}$ positions. Position $j$ for keyword $k$ has
click-through-rate $\alpha^k_j$ (possibly zero) such that 
$\alpha^k_1 \geq \alpha^k_2 \geq \hdots \geq \alpha^k_{\abs{\Gamma(k)}}$ for
each $k$.

Assuming that each keyword gets a large amount of queries, we see $\alpha^k_j$
as the sum of number of clicks that the $j$-th 
position of keyword $k$ gets across all queries that it matches. For now, let's
assume that the number of clicks a player gets in slot $j$ of keyword $k$
depends only on $j,k$ and not on the identity of the player. One is able to
relax this assumption, as we see later.

Let $\mathcal{A}_k = \{\pi_k:\Gamma(k) \hookrightarrow [\abs{\Gamma(k)}]\}$ be
the set of all allocations (one-to-one maps) from players to slots for keyword
$k$. Also, let $\Delta(\mathcal{A}_k)$ be the distributions of such allocations.
Given that, we can define the AdWords polytope in the following way: an
allocation of clicks $x$ is feasible if there is a distribution over allocations
of players to slots for each keyword such that player $i$ gets $x_i$ clicks in
expectation. More formally:

\begin{defn}[AdWords Polytope]
 The {\em AdWords polytope} is the set of feasible allocations
\footnote{Since the number of clicks is typically very large  we treat them as
divisible goods and consider also fractional allocations.} of
clicks $x =
(x_1, \hdots, x_n)$  such that there are distributions $\mathcal{D}_k \in
\Delta(\mathcal{A}_k)$ for each keyword, and
$$x_i \leq \sum_{k \in \Gamma(i)} \E_{\pi_k \sim
\mathcal{D}_k}[\alpha^k_{\pi_k(i)}]$$
\end{defn}

Our main result in this Section is that:

\begin{theorem}\label{adwords_polymatroid}
 The AdWords polytope is a polymatroid.
\end{theorem}

\comment{
The proof, which can be found in appendix \ref{setting-appendix}, uses two main
ingredients: one is the connection made by Feldman et al
\cite{feldman08} between the single-keyword sponsored search model and the
problem of scheduling in related machines. The second is a generalized version
of Hall's Theorem.
}

In order to prove the theorem, we first consider the setting with a single
keyword and all advertisers
interested in it. Let the click-through-rates be $\alpha_1 \geq \hdots \geq
\alpha_n$. 
Feldman et al \cite{feldman08} relate the problem of deciding if a vector $x$ is
feasible to a classical problem in machine scheduling, i.e.,  scheduling in
related
machines with preemptions ($Q \vert pmtn \vert C_{\max}$ \cite{graham}). What
follows is a re-statement of their characterization in a format that makes it
clear it is a polymatroidal environment.

\begin{lemma}[Feldman et al \cite{feldman08}]\label{single-keyword-lemma}
 An allocation vector $x$ is feasible iff for each $S$, $x(S) \leq
\sum_{j=1}^{\abs{S}} \alpha_j$, where $x(S) = \sum_{i \in S} x_i$ for each set
$S \subseteq [n]$.
\end{lemma}

Notice that $f_k(S) = \sum_{j=1}^{\abs{S}} \alpha^k_j$ is a submodular
function, so the set of feasible allocations for the single-keyword setting is
a polymatroid.

For the multiple-keyword setting, we say that an allocation vector $x$ is
feasible if we can write $x_i = \sum_{k \in \Gamma(i)} x_i^k$ in such a way
that the vector $(x^k_i)_{i \in \Gamma(k)}$ is feasible for keyword $k$, i.e.,
$x^k(S) \leq f_k(S)$ for every $S \subseteq \Gamma(k)$.

The fact that this allocation set is a polymatroid is a direct consequence of
the following theorem, which is a polymatroidal version of Rado's Theorem due
to McDiarmid \cite{mcdiarmid75}.

\begin{theorem}[McDiarmid \cite{mcdiarmid75}]
Given a bipartite graph $([n] \cup [m], E)$,  its neighborhood
map $\Gamma(\cdot)$, $m$ submodular functions $f_1, \hdots, f_m$ and their
respective polymatroids $P_1, \hdots, P_m$, then the set:
$$P^* = \{x \in \R^n_+; x_i = \sum_{k \in \Gamma(i)} x_i^k \text{  and  } x^k
\in P_k \} $$
is a polymatroid defined by the function $$f^*(S) = \sum_k f_k(S \cap
\Gamma(k))$$.\\
\end{theorem}

\textbf{Quality factors:} So far, we assumed that the click-through-rate of
player $i$ allocated to slot $j$ of keyword $k$  depends solely on $k$
and $j$. More generally, we would like to consider the click-through-rate of a
slot
depending also on the player allocated in that slot. Let $\alpha^k_{j,i}$ be the
click-through-rate of position $j$ of keyword $j$ when player $i$ is placed
there. Traditionally, we consider the click-through-rates in  a product form,
i.e.,
$\alpha^k_{j,i} = \alpha^k_j \cdot \gamma^k_i$ where $\gamma^k_i$ is called
\emph{quality factor}. Assuming quality factors are public information, one can,
in a similar way, define a polytope of feasible allocations. In general it will
not be a polymatroid.

If the quality factors are uniform among all queries,
i.e., $\gamma^k_i = \gamma_i$, the the set of feasible allocations is given by
$P^*_\gamma = \{x; (\frac{x_i}{\gamma_i})_i \in P^*\}$ where $P^*$ is the
AdWords polytope defined as a function of $\alpha^k_j$.
It is a scaled polymatroid, for which a variant of the auction in 
section \ref{sec:clinching_main} satisfies all desirable properties (see
appendix \ref{appendix:framework_limits}).

\section{Limitations of auctions for\columnsversion{}{\\} budget-constrained
agents}
\label{sec:characterization} \label{SEC:CHARACTERIZATION}
Previously, we argued why simple modifications to the clinching
auction would not work for polyhedral environments beyond (scaled)
polymatroids. Here, we explore the possibility 
of designing an auction of a different format achieving those properties
and  show that this is not possible even for two players. We do so
through a general characterization of Pareto-optimal auctions with desirable
properties. Before stating the characterization, we study the case with one
budget-constrained player and prove some lemmas that are useful in proving the
general characterization result later. 

\subsection{One budget-constrained player}\label{subsec:one_budgeted_player}

For ease of exposition, we first focus on $2$ players and assume that the feasible set
of allocations $P$ has a smooth and strictly-concave boundary, in the sense
that for each $v \in \R^2_+$ there is a single point $x^*(v) \in P$ maximizing
$v^t x$ such that $x^*(v)$ is a $\mathcal{C}^\infty$-function. In fact, one
can approximate any polytope by such a set using the technique of Dolev et al
\cite{linial11}. Using compactness arguments, it is possible to get an auction
for the original environment by taking the limit of the auctions obtained for
its $\mathcal{C}^\infty$-approximations.

Assume that player $1$ is not budget constrained and player $2$ has
budget $B_2$ and let $(x^*, p^*)$ be the VCG mechanism for this setting. Now,
we can define the function:
$$\xi(v_1) = \min\{v_2; p_2^*(v_1, v_2) \geq B_2\}$$

\begin{theorem}\label{thm:one_budgeted_player}
 The allocation rule $$x(v_1, v_2) = x^*(v_1, \min\{v_2, \xi(v_1)\})$$ is
monotone. Moreover, when coupled with the appropriate payment rule, it
generates a Pareto-optimal and budget feasible mechanism
\end{theorem}

\begin{proof}
 The main part of the proof is to show that the allocation is monotone. If we
show that, it is clearly budget feasible for player $2$, since we use the
VCG-payment rule until the point the budget of player $2$ gets exhausted and
from that point on, the allocation is constant. When the budgets of the players
are not exhausted, the allocation is efficient (since it mimics VCG) and
therefore is Pareto-optimal. The allocation when the budget of player $2$ is
exhausted is equivalent to the VCG allocation of a pair $(v_1, v'_2)$ with
$v'_2 \leq v_2$, so player $1$ is getting $x_1^*(v_1, v'_2) \geq x_1^*(v_1,
v_2)$ by monotonicity of VCG. This implies Pareto-optimality as a consequence
of Lemma \ref{charac-lemma}.

\emph{Monotonicity:} The allocation rule is clearly monotone for player
$2$. We need to show it is monotone for player $1$, i.e. that the function $t
\mapsto x_1(v_1+t, v_2)$ is monotone non-decreasing. It is clearly so for
intervals where $\xi(v_1+t) \geq v_2$, so let's assume that for $t \in
(-\epsilon, +\epsilon)$ we have $\xi(v_1+t) < v_2$. Our goal is to show that:
$\frac{d}{dt} x_1^*(v_1+t, \xi(v_1+t)) \geq 0$. Since the VCG-allocation lies in
the boundary of $P$, this is the same as showing that $\frac{d}{dt}
x_2^*(v_1+t, \xi(v_1+t)) \leq 0$. The crucial observation is that the
VCG-payment for player $2$ on the curve $(v_1+t, \xi(v_1+t))$ is constant, i.e.:
\columnsversion{
$$\begin{aligned} & B_2  \equiv p_2^*(v_1+t, \xi(v_1+t)) = \xi(v_1+t)
x_2^*(v_1+t, \xi(v_1+t)) -
\int_0^{\xi(v_1+t)} x_2^*(v_1+t,u) du \end{aligned}$$
}{
$$\begin{aligned} & B_2  \equiv p_2^*(v_1+t, \xi(v_1+t)) = \\ & = \xi(v_1+t)
x_2^*(v_1+t, \xi(v_1+t)) -
\int_0^{\xi(v_1+t)} x_2^*(v_1+t,u) du \end{aligned}$$
}
Now, we can simply derivate it with respect to $t$. We use
the notation $\partial_i f(\cdot)$ for the derivative of $f$ with respect to
the $i$-th variable. We also define $x_2^*(t) = x_2^*(v_1+t, \xi(v_1+t))$. Now, 
\columnsversion{
$$\begin{aligned} & 0 = \xi'(v_1+t) x_2^*(t) + \xi(v_1+t) \frac{d}{dt} x_2^*(t)
 - \xi'(v_1+t)
x_2^*(t) - \int_0^{\xi(v_1+t)} \partial_1 x^*_2(v_1+t, u) du \end{aligned}$$
$$\xi(v_1+t) \frac{d}{dt} x_2^*(t) = \int_0^{\xi(v_1+t)} \partial_1 x^*_2(v_1+t,
u) du \leq 0$$
}{
$$\begin{aligned} & 0 = \xi'(v_1+t) x_2^*(t) + \xi(v_1+t) \frac{d}{dt} x_2^*(t)
- \\ & \quad \quad - \xi'(v_1+t)
x_2^*(t) - \int_0^{\xi(v_1+t)} \partial_1 x^*_2(v_1+t, u) du \end{aligned}$$
$$\xi(v_1+t) \frac{d}{dt} x_2^*(t) = \int_0^{\xi(v_1+t)} \partial_1 x^*_2(v_1+t,
u) du \leq 0$$
}
since $x_2^*(v_1, v_2)$ decreases with $v_1$ by the definition of the VCG
allocation.
\end{proof}

A variant of the proof can be used to show the following result for $2$
budget constrained players. This is useful for our general characterization.

\begin{corollary}\label{cor:two_budgeted_players}
 If the functions $\xi_1(v_2), \xi_2(v_1)$ are such that the regions $\{v;
v_2 \geq \xi_2(v_1) \}$ and $\{v; v_1 \geq \xi_1(v_2) \}$ are disjoint, then one
can define 
$$x(v_1, v_2) = x^*(\min\{v_1, \xi_1(v_2)\}, \min\{v_2, \xi_2(v_1)\})$$
$$p_i(v) = v_i x_i(v) - \int_0^{v_i} x_i(u, v_{-i})du.$$ If $p_2(v_1,
\xi_2(v_1)) \equiv B_2$ and $p_1(\xi_1(v_2), v_2) \equiv B_1$, then $(x,p)$ is
a mechanism with the desirable properties.
\end{corollary}

The above corollary has a strong fixed-point flavour and it is tempting to
believe one could get the existence of such a mechanism from this theorem. This
is however not true, as shown in the next section. However, this result
remains useful as a tool for
searching for such mechanisms whenever they exist. For example, one can extend the
above theorem to prove the existence of the mechanisms for polyhedral
environments when $B_1$ is much larger then $B_2$.

\subsection{Characterization and impossibility}\label{subsec:impossibility}
\label{SUBSEC:IMPOSSIBILITY}

Now we discuss our main negative result: which states an impossibility of
extending the auction for polymatroids to general polyhedral environments.

\begin{theorem}[Impossibility]\label{thm:impossibility_general_polytopes}
 There is no general auction for every polyhedral environment and every pair of
budgets that satisfies the desirable properties.
\end{theorem}

We prove it in two steps: first we prove a
sequence of lemmas characterizing $2$-player auctions for polyhedral
environments satisfying all the desirable properties. Then we fix a specific
polyhedral environment and argue that no mechanism can possibly satisfy this
characterization.

First, we begin by understanding the format of an auction with the desirable
properties where the environment is a packing polytope $P \subseteq \R^2_+$.
We start by defining a family of VCG auctions.\\

{\bf VCG-family:} We say that a mechanism is in the
{\em VCG-family} if its allocation $x(v) \in P$ is such that $$x(v) \in X^*(v)
:= \text{argmax}_{x \in P} v^t x.$$ Notice that there might be more than one
such
mechanism: if $v$ is normal to an edge of the polytope,
then the entire edge is in the argmax. Nevertheless, $x(v_1+, v_2) = \lim_{v'_1
\downarrow v_1} x(v'_1, v_2)$ and $x(v_1, v_2+) = \lim_{v'_2 \downarrow v_2}
x(v_1, v'_2)$ are common for the entire family, as we see in the following
lemma. A
consequence of this fact is that the payment function might not be unique,
but $p(v_1+, v_2)$ and $p(v_1, v_2+)$ are unique. 

\begin{lemma}\label{vcg_family_limit_lemma}
Given a convex set $P \subseteq \R^2_+$ and two allocation rules $x(v),
\tilde{x}(v) \in X^*(v) := \text{argmax}_{x \in P} v^t x$, then
$x(v_1+,v_2) = \tilde{x}(v_1+,v_2)$, where  $x(v_1+, v_2) = \lim_{v'_1
\downarrow v_1} x(v'_1,v_2)$.
\end{lemma}

The proof is elementary and can be found in Appendix
\ref{appendix:missing_proofs_part_1}.
To illustrate this fact, consider the
simple case of $P = \{x \in \R^2_+; x_1 +
x_2 \leq 1\}$. Then the VCG
mechanism is well-defined for $x_1 \neq x_2$, which is, simply to allocate to
the player with the highest value the entire amount. But notice that completing
this mechanism with any allocation in the points $(v,v)$ generates a truthful
mechanism. The payments of different mechanisms of the VCG family differ on
$(v,v)$, for example $p_i(v,v) = v x_i(v)$, but notice that the payments
everywhere else are well-defined. \\

{\bf Pareto-optimal mechanisms}: Now we turn our attention back
to Pareto-optimal mechanisms for two budget-constrained players. Let $(x,p)$
be such mechanism. As a direct consequence of the characterization of Pareto
optimal outcomes (Lemma \ref{charac-lemma}), we know the following:

\begin{itemize}
 \item if $p_i(v) < B_i$ then $x_i(v) \geq \min \{x_i; x \in X^*(v) \}$
 \item if $p_1 (v) <B_1, p_2 (v) < B_2$ then $x(v) \in X^*(v)$
\end{itemize}

\noindent And a simple consequence of truthfulness:

\begin{itemize}
 \item if $p_i (v) = B_i$ then for all $v'_i \geq v_i$, $x(v'_i, v_{-i}) =
x(v_i, v_{-i})$.
\end{itemize}

Now, we are ready to start proving the characterization theorem. We will
characterize the mechanism in terms of the regions in the space of valuations
where the budgets get exhausted. For formally, we are interested in
understanding the sets:
$$E_i = \{v \in \R^2_+; p_i(v) = B_i\}$$

\begin{lemma}\label{E_curves_lemma}
 If mechanism $(x,p)$ has the desirable properties, then either $E_1 \cap E_2 =
\emptyset$ or there exists some $a \in \R^2_+$ such that $\{v; v
> a\} \subseteq E_1 \cap E_2 \subseteq \{v; v \geq a\}$.
\end{lemma}

\begin{proof}
 Assume that $E_1 \cap E_2$ is not empty. Then we will prove the lemma in two
parts. For the first part we will prove two statements: (i) if that if $v^0
\in E_1 \cap E_2$ and $v \geq v^0$ then $x(v) = x(v^0)$ and (ii) if $v, v' \in
E_1 \cap E_2$ and $v^*_i = \min\{v_i, v'_i\}$ then $x(v^*) = x(v) = x(v')$.
Then for the second part, we show that this whole region that has constant
allocation has budget exhausted for the two players. See figure
\ref{fig:lemma_E_curves_lemma} for an illustration of the proof.

For (i), let $v^1 = (v_1, v^0_2)$ and $v^2 = (v^0_1, v_2)$ and notice that
$x(v^0) = x(v^1) = x(v^2)$, since budgets are exhausted so the allocation
can't increase. By
monotonicity, $x_2(v) \geq x_2(v^1) = x_2(v^0)$ and $x_1(v) \geq x_1(v^2)  =
x_1(v^0)$. Since all allocations lie in the boundary of the polytope, we must
have $x(v) = x(v^0)$.

The proof of (ii) is very similar, define $v^0_i = \max\{v_i, v'_i\}$. Then by
(i), $x(v) = x(v^0) = x(v')$ now, by the exact same argument as above we show
that $x(v^*) = x(v^0)$.

Now, if $a_i = \inf \{v_i; v \in E_1 \cap E_2\}$, let us show that for $v > a$,
$v \in E_1 \cap E_2$. Let us show that $v \in E_2$ and then $E_1$ is analogous.
By definition, there is some $v' \in E_1 \cap E_2$ with $v'_1 < v_1$. Then
$(v'_1, v_2) \in E_2$ since the allocation is constant, for $v > a$, the budget
of $2$ is exhausted in $(v'_1, v_2)$ iff it is exhausted in $(v'_1, v'_2)$.
Now, note that by monotonicity $x_1(v'_1, u) \leq x_1(v_1, u)$ and since
allocation is in the boundary of the polytope  $x_2(v'_1, u) \geq x_2(v_1, u)$.
By the payment formula, $p_2(v) = v_2 x_2(v) - \int_0^{v_2} x_2(v_1, u) du$
and by the fact that $x_2(v) = x_2(v'_1, v_2)$, we have that $B_2 \geq p_2(v)
\geq p_2(v'_1, v_2) = B_2$, so $v \in E_2$.
\end{proof}

\begin{figure*}
\centering
\includegraphics{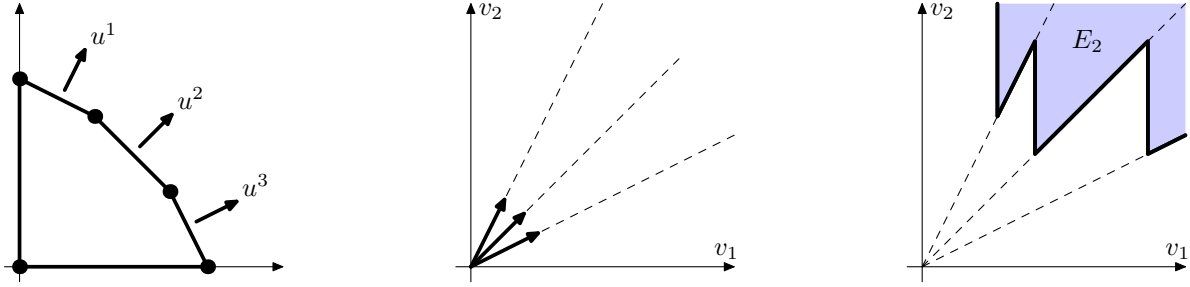}
\caption{Illustration of the proof of Lemma \ref{E_curves_lemma}}
\label{fig:lemma_E_curves_lemma}
\end{figure*}

The next lemma further describes the regions $E_1$ and $E_2$:

\begin{lemma}\label{lemma_xi_curves}
If $P$ is a packing polytope, there is a finite set of vectors $\{u^1, u^2,
\hdots, u^k\}$ such that for $v \neq t u^i$ for some $t \geq 0$, the VCG family
is uniquely defined. Moreover, if $(x,p)$ is a mechanism with the desirable
properties and $E_1 \cap E_2 \subseteq \{v; v \geq a\}$ then if
$\xi_i(v_{-i}) = \inf\{v_i; p_i(v_i, v_{-i}) = B_i\}$ and $v_{-i} < a_{-i}$,
then $(\xi_i(v_{-i}),v_{-i}) = t u^i$ for some $t,i$.
\end{lemma}

\begin{figure*}
\centering
\includegraphics{lemma_charac_21.mps}
\caption{Illustration of the proof of Lemma \ref{lemma_xi_curves}}
\label{fig:lemma_xi_curves}
\end{figure*}

\begin{proof}
The set of vectors $\{u^1, u^2,\hdots, u^k\}$ is simply the set of normals of
the edges of the polytope as depicted in the first part of 
Figure~\ref{fig:lemma_xi_curves}. If $v$ is not an edge in the polytope, then
the
point in $P$ maximizing $v^t x$ is a vertex and therefore uniquely defined. If
we draw the lines $u^i \cdot t$ for $t>0$ we divide the space of all possible
valuations in regions (see second part of the figure): the
regions correspond to the vertices and the lines to edges of the polytope. 

Now, given a certain mechanism $(x,p)$, suppose that $v_1 < a_1$ the vector
$(v_1, \xi_2(v_1))$ is not
 normal to any edge of $P$. Then clearly $x^*(v_1,
\xi_2(v_1))$ is well-defined and moreover, for some $\delta > 0$ and $v'_2 \in
[\xi_2(v_1)-\delta, \xi_2(v_1)+\delta]$, $x^*(v_1, v'_2)$ is well-defined and
constant in the $v'_2$-range. Also $(v_1, v'_2) \notin E_1 \cap E_2$, since
$v_1 < a_1$. For such $v'_2 > \xi_2(v_1)$, we
know that $(v_1, v'_2) \in E_2$,  thus it is not in $E_1$ and therefore
$x_2(v_1, v'_2) = x_2(v_1,\xi_2(v_1)) \leq x_2^*(v_1,\xi_2(v_1))$. For $v'_2 <
\xi_2(v_1)$, $(v_1, v'_2) \notin
E_2$ so  $x_2(v_1, v'_2) \geq x_2^*(v_1,v'_2)$. Using that $x_2(v_1, v'_2)$ is
constant in $v'_2$ in this range and taking $v'_2 \uparrow \xi_2(v_1)$, we get:
$x_2(v_1, v'_2) \geq x_2^*(v_1,\xi_2(v_1))$.

Now, by monotonicity, 
we have $x_2(v_1, v'_2) = x_2^*(v_1,\xi_2(v_1))$ for all $v'_2$ in the interval
$(\xi_2(v_1)-\delta, \xi_2(v_1)+\delta)$. Therefore the budget of player $2$
could not have been exhausted on $\xi_2(v_1)$.
\end{proof}

The third part of Figure \ref{fig:lemma_xi_curves} illustrates how the region
$E_2$ typically looks like. If $u^1, u^2, \hdots, u^k$ are sorted such that
$u^1$ corresponds to the edge that is higher to the left and $u^k$ corresponds
to the edge that is lower to the right, then we can divide $[0,a_1)$ in
segments $[0, r_1), [r_1, r_2), \hdots, [r_k, a_1)$ such that for $v_1 \in
[0,r_1)$, $\xi_2(v_1) = \infty$, and for $v_1 \in [r_i, r_{i+1})$, $(v_1,
\xi_2(v_1))$ is in the line $\{t \cdot u^i; t \geq 0\}$. To see why this is
true, consider $v_1 < v'_1$ and assume that $(v_1,
\xi_2(v_1))$ is in the line $\{t \cdot u^i; t \geq 0\}$. Then   $(v'_1,
\xi_2(v'_1))$ cannot be strictly above this line, by a similar argument 
 used in Lemma \ref{E_curves_lemma}: look at the allocation curves $x_2(v_1,
v_2) \geq x_2(v'_1, v_2)$ and the payment formula -- then player $2$ must be
paying just above the  $\{t \cdot u^i; t \geq 0\}$ for $v'_1$ line at least as
much as he was paying above this line for $v_1$ and hence his budgets must be
exhausted.

\subsection{Proof of the Impossibility
Theorem}

Now, we are ready to prove Theorem \ref{thm:impossibility_general_polytopes},
which states that there is no
general auction with all the desirable properties for all polyhedral
environments $P$. We fix the following setting:  a set of feasible allocations
$$P
= \{x \in \R^2_+; 2x_1 + x_2 \leq 6, x_1 + 2x_2 \leq 6\}$$ and budgets $B_1 =
B_2
= 1$. Assume that $(x,p)$ is a mechanism with the desirable properties for this
setting. We will use the characterization lemmas in Section
\ref{subsec:impossibility} to find a contradiction. We illustrate
the
flow of the proof in Figure \ref{fig:impossibility} (a brief summary of the
proof is given in the caption of the figure).

\begin{figure*}
\centering
\includegraphics{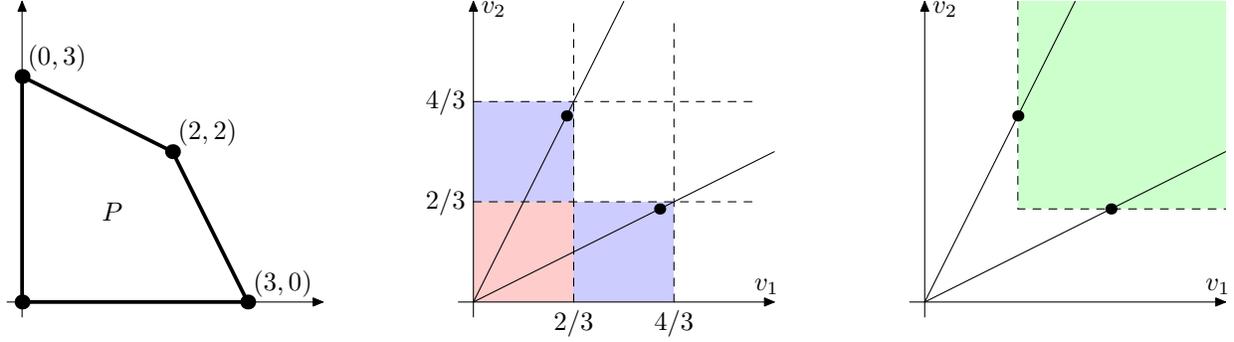}
\caption{Illustration of the proof of the impossibility theorem: in the first
part we represent the polytope $P$, in the second we represent the first main
steps: we show that any auction must resemble VCG in the red region (Fact
\ref{first_fact}) and extend the definition of the auction to the blue region
(Facts \ref{second_fact} and \ref{third_fact}), showing that the budget of
player $1$ must get exhausted at the point $(\frac{1.2381}{2},1.2381)$ and the
budget of player $2$ must get exhausted at the point
$(1.2381,\frac{1.2381}{2})$. We use this fact to show that the allocation must
be constant in the green region (Fact \ref{last_fact}), contradicting
Pareto-optimality for allocation of the form $(v_1, v_2)$ where
$\frac{1.2381}{2}< v_1 < \frac{2}{3}$ as $v_2 \rightarrow \infty$.}
\label{fig:impossibility}
\end{figure*}

\begin{fact}\label{first_fact}
In the region $[0,\frac{2}{3})\times [0,\frac{2}{3})$, the mechanism $(x,p)$
produces an efficient allocation (i.e., it is equal to some mechanism of the VCG
family).
\end{fact}

\begin{proof}
First notice that no mechanism in the VCG-family for $P$ exhausts
the budget in $[0,\frac{2}{3})\times [0,\frac{2}{3})$. Now, we turn our
attention to the mechanism $(x,p)$, which is a mechanism with the desirable
properties for this setting.

In $[0,\frac{1}{3})\times [0,\frac{2}{3})$, player $1$ cannot exhaust his
budget,
since $x_1 \leq 3$ and $v_1 < \frac{1}{3}$. We claim that in this area the
mechanism needs to behave like VCG. If there is a point in this region
where $x(v)$ doesn't maximize $v^t x$ for $x \in P$, then the budget of player
$2$ must be exhausted. So, there is
some $(v_1, v_2)$ such that for $v'_2 > v_2$, $p_2(v_1, v'_2) = B_2$ and for 
$v'_2 < v_2$, $x(v_1, v'_2)$ is in the VCG family. Now, notice that the
allocation for $v'_2 > v_2$ must be $x_2(v'_2, v_1) \leq x_2^*(v'_2, v_1)$.
This contradicts the fact that the budget is exhausted for $v'_2 \downarrow
v_2$. So, this shows that $(x,p)$ must allocate efficiently on
$[0,\frac{1}{3})\times [0,\frac{2}{3})$. Now, for $[0,\frac{2}{3}) \times
[0,\frac{1}{3})$, we can do the same argument. Now, what remains are the points
in $[\frac{1}{3},\frac{2}{3}) \times [\frac{1}{3},\frac{2}{3})$. Let $v$ be such
a point. Notice that for $(\frac{1}{3} - \epsilon, v_2)$ and $(v_1,
\frac{1}{3}-\epsilon)$, the allocation must be $(2,2)$ because of the previous
argument. By monotonicity, $x(v) = (2,2)$ which is the efficient allocation.
\end{proof}

Now, we know how any mechanism $(x,p)$ with the desirable properties should
look like in the red region of Figure \ref{fig:impossibility}. Next, we try to
understand how it should look in the blue region. In order to do so, we need
some definitions. From Lemma \ref{E_curves_lemma} we know that there is $a \in
\R_+^2 \cup \{ (\infty,\infty) \})$ such that: $ \{v; v > a\} \subseteq E_1
\cap E_2 \subseteq \{v; v \geq a\}$. We also define:
 $$\tilde{v}_2 = \min \{v_2; p_1(\frac{v_2}{2} +, v_2) = 1 \}$$
 $$\tilde{v}_1 = \min \{v_1; p_2(v_1, \frac{v_1}{2} +)  = 1\}$$
And we focus on the region $R$, which we define as the interior of the
rectangle  between $(0,0)$ and $(\tilde{v}_2 / 2, \tilde{v}_2)$.

\begin{fact}\label{second_fact}
The budget of player $1$ does not get exhausted in $R$. Also, $\frac{2}{3} < 
\tilde{v}_2 \leq 1.2381$
\end{fact}

\begin{proof}
 First, assume $\tilde{v}_2 < 4/3$. Then by Fact \ref{first_fact}, the
budget of player $1$ does not get exhausted in
$[0,\frac{2}{3})\times[0,\frac{2}{3})$ and by the definition of $\tilde{v}_2$
and lemma \ref{lemma_xi_curves} it cannot be exhausted for
$[0,\frac{\tilde{v}_2}{2})\times(\frac{1}{3}, \tilde{v}_2)$. Notice that $a$
cannot be in this region, because it would also contradict the definition of
$\tilde{v}_2$.

Given this fact, let us analyze how the mechanism should be in this setting. We
do so, by fixing $v_1$ and looking at $x_2(v_1, v_2)$. We can use the same
argument as in the previous fact to argue that $x(v_1, v_2)$ must be the
efficient allocation for $v_2 < 2 v_1$ in $R$ and also for all $v_1 \leq
\frac{1}{3}$. For $v_1 > \frac{1}{3}$ and $v_2 > 2 v_1$, if we allocate as in
VCG,
we exceed the budget. So, the allocation for those points must exactly match
the budget of player $2$. So, the only possible value must be such that:
$$2 \cdot \frac{1}{2} v_1 + (x_2(v) - 2) 2 v_1 = 1$$
and therefore the allocation must be $$x(v) = \left( 3-\frac{1}{v_1},
2+\frac{1-v_1}{2 v_1} \right).$$

Now, this determines the payment of player $1$ in the rectangle. We know that
for $\delta
\downarrow 0$, $p_1(\frac{\tilde{v}_2}{2} -\delta, \tilde{v}_2 - 3\delta) < 1$,
which we can write as:
\columnsversion{
$$\begin{aligned} \lim_{\delta\downarrow 0 } p_1( \frac{\tilde{v}_2}{2} -\delta,
\tilde{v}_2 - 3\delta ) & = 2 \cdot \frac{\tilde{v}_2}{2} -
\int_{1/3}^{\tilde{v}_2/2} 3 - \frac{1}{z}
dz = 1 - \frac{\tilde{v}_2}{2} + \log \left( \frac{3}{2} \tilde{v}_2
\right)
\leq 1 \end{aligned}$$
}{
$$\begin{aligned} \lim_{\delta\downarrow 0 } p_1( \frac{\tilde{v}_2}{2} -\delta,
\tilde{v}_2 - 3\delta ) & = 2 \cdot \frac{\tilde{v}_2}{2} -
\int_{1/3}^{\tilde{v}_2/2} 3 - \frac{1}{z}
dz = \\ & = 1 - \frac{\tilde{v}_2}{2} + \log \left( \frac{3}{2} \tilde{v}_2
\right)
\leq 1 \end{aligned}$$
}
This implies that $\tilde{v}_2 \leq 1.2381$. Notice that this excludes the fact
that $\tilde{v}_2 > 4/3$, otherwise we could have done the same analysis on $[0,
\frac{2}{3}] \times [0,\tilde{v}_2]$ and arrived in the same conclusion that
$\tilde{v}_2 \leq 1.2381$.
\end{proof}

\begin{fact}\label{third_fact}
$\tilde{v}_2 = 1.2381$ and $x(\frac{\tilde{v}_2}{2}+, \tilde{v}_2) = (2,2)$.
\end{fact}

\begin{proof}
Since  $\frac{2}{3} < \tilde{v}_2 \leq 1.2381$, we know that
$x_2(\frac{\tilde{v}_2}{2}+, \tilde{v}_2) \geq 2$ by using that
$x_2(\frac{\tilde{v}_2}{2}+, \frac{1}{3}) = 2$ (Fact \ref{first_fact}) and
monotonicity. Therefore, $x_1(\frac{\tilde{v}_2}{2}+, \tilde{v}_2) \leq 2$.
Writing the payment for player $1$ at this point we get:
\columnsversion{
$$\begin{aligned} 1 & = p_1(\frac{\tilde{v}_2}{2}+, \tilde{v}_2) =
x_1(\frac{\tilde{v}_2}{2}+,
\tilde{v}_2) \cdot \frac{\tilde{v}_2}{2} -
\int_{1/3}^{\tilde{v}_2/2} 3 - \frac{1}{z}
dz \leq 1 - \frac{\tilde{v}_2}{2} + \log \left( \frac{3}{2}
\tilde{v}_2 \right) \end{aligned}$$
}{
$$\begin{aligned} 1 & = p_1(\frac{\tilde{v}_2}{2}+, \tilde{v}_2) =
x_1(\frac{\tilde{v}_2}{2}+,
\tilde{v}_2) \cdot \frac{\tilde{v}_2}{2} -
\int_{1/3}^{\tilde{v}_2/2} 3 - \frac{1}{z}
dz \leq \\ & \leq 1 - \frac{\tilde{v}_2}{2} + \log \left( \frac{3}{2}
\tilde{v}_2 \right) \end{aligned}$$
}
which implies that $\tilde{v}_2 = 1.2381$ and $x(\frac{\tilde{v}_2}{2}+,
\tilde{v}_2) = (2,2)$, since all inequalities must be tight.
\end{proof}

\begin{fact}\label{last_fact}
$\tilde{v}_1 = \tilde{v}_2 = 1.2381$ and $x(v) =
(2,2)$ for all valuation profiles  $v > (\frac{\tilde{v}_1}{2},
\frac{\tilde{v}_2}{2})$.
\end{fact}

\begin{proof}
 We can apply the same argument exchanging $1$ and $2$ and conclude that
$\tilde{v}_1 = \tilde{v}_2$. Now, to see that for $v'_1, v'_2 >
\frac{\tilde{v}_i}{2}$ we have $x(v'_1, \tilde{v}_2) = x(\tilde{v}_1, v'_2) =
(2,2)$, we analyze four regions. If $v \in
(\tilde{v_i}/2, \tilde{v_i}] \times (\tilde{v_i}/2, \tilde{v_i}]$ or $v \in
[\tilde{v_i}, \infty) \times [\tilde{v_i}, \infty)$ we can use the standard
monotonicity argument to show that $x(v) = (2,2)$.

For the regions $[\tilde{v}_1, \infty) \times (\frac{\tilde{v}_1}{2},
\tilde{v}_1)$ and $(\frac{\tilde{v}_2}{2}, \tilde{v}_2) \times [\tilde{v}_2,
\infty)$ is a little trickier. We do the analysis for the first one. The second
is analogous.

Clearly $p_2(\tilde{v}_1,\tilde{v}_2) = 1$.
Now, for $v'_1 > \tilde{v}_1$, $x(v'_1, \tilde{v}_2) = x(\tilde{v}_1,
\tilde{v}_2) = (2,2)$. And for all $v_2$ we have $x_1(v'_1, v_2) \geq
x_1(\tilde{v}_1,v_2)$ and therefore $x_2(v'_1, v_2) \leq
x_2(\tilde{v}_i,v_2)$. Now, since $p_2(v_1, v_2) = v_2 x_2(v) - \int_0^{v_2}
x_2(v_1, u) du$ clearly $p_2(v'_1, \tilde{v}_2) \geq p_2(\tilde{v}_1,
\tilde{v}_2) = 1$. Therefore $p_2(v'_1, \tilde{v}_2) = 1$, since the mechanism
respects budgets. Notice that the only way it can be true is that if $x_2(v'_1,
v_2) = x_2(\tilde{v}_i,v_2)$, so we must have $x_2(v'_1,
v_2) = (2,2)$ for $v'_1 \in [\tilde{v}_i, \infty)$.

We can use the exact same argument for region: $(\frac{\tilde{v}_i}{2},
\tilde{v}_i) \times [\tilde{v}_i, \infty) $.
\end{proof}

Now we are ready to prove Theorem \ref{thm:impossibility_general_polytopes}:

\begin{proofof}{Theorem \ref{thm:impossibility_general_polytopes}}
Now, by putting Fact \ref{last_fact} and Fact \ref{first_fact} together, we get
a contradiction with the Pareto-optimality:  consider $\frac{\tilde{v}_i}{2}
< v_1 <
\frac{2}{3}$ then by facts \ref{last_fact} and \ref{first_fact} combined, we
know that $x_2(v_1, v_2) = 0$ for $v_2 < \frac{v_1}{2}$ and $x_2(v_1, v_2) = 2$
for $v_2 > \frac{v_1}{2}$, so the budget of player $2$ never gets exhausted
even for $v_2 \rightarrow \infty$. This contradicts
Pareto-optimality for $v_2 > 2 v_1$, since if his budget is not exhausted, he
should get allocated at least as much as he gets in VCG.
\end{proofof}

\subsection{Multi-unit auctions with decreasing\columnsversion{}{\\}
marginals}\label{subsec:decreasing marginals}

As a by-product of Theorem \ref{thm:impossibility_general_polytopes}, we can
answer in a negative way the question of the existence of truthful
Pareto-optimal auctions for multi-unit auctions with decreasing marginals.
Consider the following setting:\\

\noindent \textbf{Setting:} Consider a supply of $s$ of a certain divisible good
and
two players in such a way that the feasible allocations are $(x_1, x_2)$ such
that $x_1 + x_2 \leq s$. Player $i$ has a public budget $B_i$ and a
private valuation which is a increasing concave function $V_i:[0,s]
\rightarrow \R_+$. Upon getting $x_i$ units of the good and paying $p_i$, player
$i$ has utility $u_i = V_i(x_i) - p_i$.\\

It is tempting to believe that one could adapt the clinching framework in
Algorithm \ref{polyhedral-clinching-auction} to deal with this setting, by
simply redefining the demand function as something like:
$$d_i = \min \left\{\frac{B_i}{p}, \max \{x_i; \partial V_i(\rho_i+x_i)
\leq p \} \right\}$$
where $\partial V_i(x_i)$ is the marginal valuation at $x_i$. Indeed, if
$V_i(x_i) = v_i \cdot x_i$, this recovers the original way of calculating
demands. In appendix \ref{appendix:clinching-counter-example} we give a
counter-example showing that the clinching framework with this new demand
function results in a non-truthful mechanism. The intuition behind the counter
example is that some player can increase his declared value on items he won't
get anyway in order to increase the payment of his opponent, exhausting his
budget earlier. This way, he is able to get items for cheaper in the end.

In the following theorem, we show that no auction mechanism can satisfy all the
desirable properties for this setting:

\begin{theorem}
There is no truthful, Pareto-optimal and budget-feasible auction for this
setting.
\end{theorem}

\begin{proof}
 Suppose that $(\hat{x}(V_1,V_2), \hat{p}(V_1, V_2))$ is a mechanism satisfying
all the desirable properties for multi-unit auctions with decreasing marginals.
Then we can use it as a black-box to construct a mechanism for a general
polyhedral environment, contradicting Theorem
\ref{thm:impossibility_general_polytopes}.  Given a certain polyhedral
environment, we can describe $$P = \{x \in
[0,\alpha] \times [0, \beta]; x_2 \leq h(x_1)\}$$ where $h:[0,\alpha]
\rightarrow
[0, \beta]$ is a monotone non-increasing concave function, $h(0) = \beta$ and
$h(\alpha) = 0$. Now, using $(\hat{x}, \hat{p})$ for $s=1$, build the following
mechanism:
if players report valuations $v_1, v_2$ build the following concave functions:
$V_1(x_1) = v_1 \cdot \alpha x_1$ and $V_2(x_2) = v_2 \cdot h(\alpha -
\alpha x_2)$.

Now, simply define $$\begin{aligned} & x(v_1, v_2) = (\alpha \hat{x}_1,
h(\alpha \hat{x}_1)) \\  & p(v_1, v_2) =
\hat{p}(V_1, V_2).\end{aligned}$$ The mechanism is clearly truthful,
individually
rational and budget-feasible. It is also easy to see that sets of allocation
can me mapped 1-1 between those two settings, preserving Pareto-optimality.
\end{proof}

\bibliographystyle{abbrv}
\bibliography{sigproc}  
%
%
\appendix
\section{Extensions and Limitations of\columnsversion{}{\\} the  clinching
framework}\label{appendix:framework_limits}

\textbf{Scaled polymatroids:} If $P \subseteq \R^n_+$ is a
polymatroid, $\gamma \in \R^n_+$, then we call
$P_\gamma = \{x; (\frac{x_i}{\gamma_i})_i \in P\}$ a scaled polymatroid. If the
environment is $P_\gamma$, it is easy to see that a truthful, individually-rational, 
Pareto-optimal and budget-feasible auction is obtained by running the
polymatroid clinching auction on $P$ with inputs $\gamma_i v_i$ instead of
$v_i$. It is simple to see that this is equivalent to a standard clinching
auction with input values $v_i$ but price clocks advancing on a different speed
for each player. Scaled polymatroids are important since they correspond to the
setting of AdWords with Quality Factors discussed at the end of
section \ref{subsec:adwords_polytope}).\\

\textbf{Beyond scaled polymatroids:} one could change the way
clinching is done and one could change the price trajectories, maybe in a more
sophisticated way then the one we did for the scaled polymatroids. If the
trajectory is such that it only depends on $P$ and budgets (not on values) and
$p_i$ never decreases, the auction retains truthfulness and budget
feasibility. Here we
argue that none of such changes would generate a Pareto-optimal auction when $P$
is not a scaled polymatroid. Assume
$n=2$ for simplicity and imagine that there is trajectory for the price vector $p$
and a clinching procedure. Also assumes valuations are much smaller than
budgets in such a way that the mechanism cannott exhaust budgets and therefore
the auction must allocate like VCG. Any such mechanism must decide on the whole
allocation
the first time $p_i = v_i$ for some component, having only the information
that $v_j \geq p_j$ for the other component, since at this point he needs to
allocate to the second player. So, the environment must be such that there is a
price trajectory $p(t)$, where the optimal allocation is constant for all
points $(p_1, v_2)$ for all $v_2 > p_2$. And also, it should be constant for
all $(v_1, p_2)$ for all $v_1 > p_1$. Notice that scaled polymatroids are
exactly those environments.

\section{Faster clinching sub-routines}\label{appendix:eff_clinching_adwords}

In Lemma \ref{eff-clinching} we showed that for any generic polymatroidal
environment $P$ we can perform the clinching step in polynomial time if we have
oracle access to the submodular function defining the polymatroid. In order to
do so, we solve a submodular minimization problem. For most practical
applications, however, one can design much simpler and faster algorithms for
clinching. Clinching involves solving the following problem: given an
environment $P$, $\rho \in P$ and $d \in \R^n_+$ we want to compute:

\begin{equation}\label{opt-problem-clinching}
 \max \{\one^t x; x + \rho \in P, 0 \leq x \leq d\}
\end{equation}

We illustrate how to solve this problem efficiently for the single-keyword AdWords
polytope: given $\alpha_1 \geq \alpha_2 \geq \hdots \geq \alpha_n$, consider
the environment:
$$P = \{x \in \R^n_+; x(S) \leq \textstyle\sum_{j=1}^{\abs{S}} \alpha_j, \forall
S
\subseteq [n] \}$$

\begin{lemma}
 For the single-keyword AdWords polytope, the optimization problem defined in
equation (\ref{opt-problem-clinching}) can be solved using the following greedy
algorithm: we can assume wlog that the
components are sorted such that $\rho_1 + d_1 \geq \rho_2 + d_2 \geq \hdots \geq \rho_n
+ d_n$. Now, define inductively $$z_i = \min\{\rho_i + d_i,
\textstyle\sum_{j=1}^i \alpha_j
- \textstyle\sum_{j=1}^{i-1} z_j\}.$$ Then $\sum_i z_i - \rho_i$ is the solution
to the
problem.
\end{lemma}

\begin{proof}
If we drop the restriction that $x \geq 0$ in  (\ref{opt-problem-clinching}),
then it is easy to see that $x = z-\rho$ is an optimal solution to this problem
using, for example, a local exchange argument. Now, we show that we can fix
this problem, by modifying $z$ such that $z \geq \rho$.\comment{
 It is easy to see that $z$ satisfies $z \in P$ and $z \leq \rho+d$. This
implies that $\one^t z - \one^t \rho$ is an
upper bound to (\ref{opt-problem-clinching}). In order to show it is also a
lower bound, we can define $y = z-\rho$. Clearly $y+\rho \in P$ and $y \leq d$.
Then $y$ is almost feasible w.r.t the constraints in
(\ref{opt-problem-clinching}) except that $y_i$ might be negative.}

In order to fix that, consider the smaller $i$ such that $z_i < \rho_i$. By the
definition of $z_i$, it must be the case that $z_i = \sum_{j=1}^i \alpha_j
- \sum_{j=1}^{i-1} z_j$, so $\sum_{j=1}^i z_j = \sum_{j=1}^i \alpha_j$. Then
there must be some $k < i$ such that $z_k > \rho_k$, otherwise we would have
$\sum_{j=1}^i \rho_j > \sum_{j=1}^i \alpha_j$ contradicting the fact that $\rho
\in P$. Notice we can increase $z_i$ by some small $\delta$ and decrease $z_k$
by a small $\delta$. And obtain another vector $z$ which is also such that $z
\in P, z \leq \rho+d$ and has the same $\one^t z$ value (the fact that $z \in
P$ after this transformation is due to the nature of the constraints). We can
repeat this process until we get $z \geq \rho$. 
\end{proof}

\section{Missing proofs in Section
\ref{SEC:CHARACTERIZATION} }\label{appendix:missing_proofs_part_1}

\begin{proofof}{Lemma \ref{vcg_family_limit_lemma}}
Suppose that $x(v_1+,v_2) \neq \tilde{x}(v_1+,v_2)$, then say that
$x_1(v_1+,v_2) > \tilde{x}_1(v_1+,v_2)$. Since both are in the boundary of the
polytope, it means that for all $(v'_1,v_2)$ with $v'_1 > v_1$, 
$\tilde{x}_1(v'_1, v_2) \geq x_1(v_1+,v_2)$, since the points to the right of
$x(v_1+,v_2)$ are clearly better then the ones to the left of it. So
$\tilde{x}_1(v'_1, v_2)$ can't converge to
$\tilde{x}_1(v_1+,v_2) < x_1(v_1, x_2)$.
\end{proofof}

\section{Counter-example to clinching with decreasing marginals}
\label{appendix:clinching-counter-example}

Consider the setting of multi-unit auctions with decreasing marginals described
in section \ref{subsec:decreasing marginals}, and the variant of the
clinching auction in Algorithm \ref{polyhedral-clinching-auction} where demands
are calculated by:
$$d_i = \min \left\{\frac{B_i}{p}, \max \{x_i; \partial V_i(\rho_i+x_i) \leq
p \} \right\}$$

Our main claim  here is that this auction is not truthful. We show that by providing
an example: consider an initial supply of $s=2$ budgets $B_1 = \infty, B_2 = 4$
and valuations:

$$V_1(x) = \left\{ \begin{aligned} 4x, \quad & x \in [0,1] \\ 4+x, \quad  & x
\in [1,2] \end{aligned} \right.  \qquad V_2(x) = 3x, \quad x \in [0,2]$$

The outcome of that is $x_1 = 1, x_2 = 1$ and $p_1 = 3, p_2 = 1$ since budgets
never get exhausted. Now, we show that player $1$ has a profitable deviation.
He could report the following valuation instead:
$$\tilde{V}_1(x) = \left\{ \begin{aligned} 4x, \quad & x \in [0,1] \\
4+2x, \quad & x \in [1,2] \end{aligned} \right.$$
The the price clock increases up to $2$ without any clinching. For
$2+$, the demand of player $1$ drops and player $2$ is able to clinch one unit
by the price of $2$, remaining with budget $2$. Now, as price goes up, player
$2$'s demand keeps decreasing, since his leftover budget is only $2$ and the
price is greater than $2$. As a result, player $1$ will be able to start
clinching at price $p=2+$ and therefore will have $\tilde{x}_1 = 1$,
$\tilde{p}_1 < 3$, which is a strict improvement over the truth-telling strategy.

\columnsversion{}{
\balancecolumns
}
\end{document}